\documentclass[a4paper,onecolumn,superscriptaddress,10pt,unpublished,nopdfoutputerror]{compositionalityarticle}

\usepackage[utf8]{inputenc}
\usepackage[english]{babel}
\usepackage[T1]{fontenc}
\usepackage{amsmath}
\usepackage{hyperref}
\usepackage{tikz}
\usepackage{lipsum}
\usepackage{tikz-cd}
\usepackage{mathrsfs}

\usepackage{physics}
\usepackage{stmaryrd}
\usepackage{tikz}
\usepackage{tikz-cd}
\usepackage{tikzit}
\usepackage{circuitikz}
\usepackage{tubes}
\usepackage{xsavebox}

\input{styles.tikzdefs}

\tikzstyle{discarding}=[fill=white, draw=black, shape=circle, style=upground]
\tikzstyle{smalldiscarding}=[fill=white, draw=black, style=upground, scale=0.75]
\tikzstyle{backdiscard}=[fill=white, draw=black, shape=circle, style=downground, scale=0.5]
\tikzstyle{smallbackdiscard}=[fill=white, draw=black, shape=circle, style=downground, scale=0.5]
\tikzstyle{state}=[fill=white, draw=black, style=triang, tikzit shape=rectangle]
\tikzstyle{kstate}=[fill=white, draw=black, style=kpoint, tikzit shape=rectangle]
\tikzstyle{kstateconj}=[fill=white, draw=black, style=kpoint conjugate, tikzit shape=rectangle]
\tikzstyle{kstateBIG}=[fill=white, draw=black, style=big kpoint, tikzit shape=rectangle]
\tikzstyle{effect}=[fill=white, draw=black, style=triangdag]
\tikzstyle{keffect}=[fill=white, draw=black, style=kpoint adjoint]
\tikzstyle{keffectconj}=[fill=white, draw=black, style=kpoint transpose]
\tikzstyle{morphdag}=[style=mapdag]
\tikzstyle{morph}=[style=hadamard]
\tikzstyle{WIDEmorph}=[style=hadamard, minimum width=14mm]
\tikzstyle{morphtrans}=[style=maptrans]
\tikzstyle{morphconj}=[style=mapconj]
\tikzstyle{CPMmorph}=[style=dmap]
\tikzstyle{CPMmorphconj}=[style=dmapconj]
\tikzstyle{CPMmorphdag}=[style=dmapdag]
\tikzstyle{CPMmorphtrans}=[style=dmaptrans]
\tikzstyle{CPMstate}=[fill=white, draw=black, style=triang, doubled]
\tikzstyle{CPMstateBIG}=[fill=white, draw=black, style={triang_lesssep}, doubled]
\tikzstyle{CPMkstate}=[fill=white, draw=black, style=kpoint, tikzit shape=rectangle, doubled]
\tikzstyle{CPMkstateconj}=[fill=white, draw=black, style=kpoint conjugate, tikzit shape=rectangle, doubled]
\tikzstyle{CPMkstateBIG}=[fill=white, draw=black, style=big kpoint, tikzit shape=rectangle, doubled]
\tikzstyle{CPMkeffect}=[fill=white, draw=black, style=kpoint adjoint, doubled]
\tikzstyle{CPMkeffectconj}=[fill=white, draw=black, style=kpoint transpose, doubled]
\tikzstyle{UHfB}=[fill=white, draw=black, style=triangdag, doubled, inner sep=-2pt]
\tikzstyle{leak}=[style=tinypoint, regular polygon rotate=-90]
\tikzstyle{leakfill}=[style=tinypoint, regular polygon rotate=-90, fill=black]
\tikzstyle{Z}=[style=dot, fill=green]
\tikzstyle{X}=[style=dot, fill=red]
\tikzstyle{black_dot}=[style=dot, fill=black]
\tikzstyle{white_dot}=[style=dot, fill=white]
\tikzstyle{qblack_dot}=[style=ddot, fill=black]
\tikzstyle{qwhite_dot}=[style=ddot, fill=white]
\tikzstyle{whitephase}=[style=wphase dot, fill=white]
\tikzstyle{qredphase}=[style=phase dot, fill=red]
\tikzstyle{qgreenphase}=[style=phase dot, fill=green]
\tikzstyle{had}=[style=hadamard, doubled]
\tikzstyle{box}=[style=hadamard]
\tikzstyle{bigbox}=[style=hadamard, minimum height=4mm, minimum width=8mm]
\tikzstyle{classhad}=[style=hadamard]
\tikzstyle{antipode}=[style=anti]

\tikzstyle{dottededge}=[-, dotted]
\tikzstyle{double_edge}=[-, double]
\tikzstyle{double_dot}=[-, double, very thick, dotted]
\tikzstyle{double_bold}=[very thick, double]
\tikzstyle{double edge}=[-, very thick]
\tikzstyle{new edge style 0}=[<-]
\tikzstyle{new edge style 1}=[-, draw={rgb,255: red,242; green,233; blue,206}, fill={rgb,255: red,242; green,233; blue,206}]
\tikzstyle{morphism_shade}=[-, draw=black, fill={rgb,255: red,242; green,233; blue,206}, line join=bevel]
\tikzstyle{supermap_shade}=[-, fill={rgb,255: red,216; green,215; blue,242}, draw=black, line join=bevel]
\tikzstyle{hole_shade}=[-, fill=white, draw=black,line join=bevel]
\tikzstyle{new edge style 2}=[-, draw={rgb,255: red,14; green,188; blue,83}]
\tikzstyle{new edge style 3}=[<-, draw={rgb,255: red,234; green,209; blue,255}]
\tikzstyle{new edge style 4}=[<-, draw={rgb,255: red,0; green,128; blue,128}]
\tikzstyle{new edge style 5}=[-, draw={rgb,255: red,214; green,110; blue,62}]
\tikzstyle{new edge style 6}=[-, draw={rgb,255: red,174; green,20; blue,174}]

\newcommand{\tikzfigscale}[2]{\scalebox{#1}{\input{#2.tikz}}}

\newcommand{\cat}[1]{\mathcal{#1}}

\newcommand{\seq}{\varogreaterthan}

\DeclareFontFamily{U}{min}{}
\DeclareFontShape{U}{min}{m}{n}{<-> udmj30}{}

\makeatletter
\newcommand{\xRightarrow}[2][]{\ext@arrow 0359\Rightarrowfill@{#1}{#2}}
\makeatother

\usepackage{comment}

\newcommand{\bigcircle}{\mathop{\mathchoice
  {\vcenter{\hbox{\LARGE$\bigcirc$}}}
  {\vcenter{\hbox{\Large$\bigcirc$}}}
  {\vcenter{\hbox{\large$\bigcirc$}}}
  {\vcenter{\hbox{\small$\bigcirc$}}}
}}

\pgfsetlayers{background,strings,edgelayer,nodelayer,foreground,main}

\usepackage{amssymb}

\usepackage{amsthm}

\theoremstyle{plain}
\newtheorem{theorem}{Theorem}[section]
\newtheorem{lemma}{Lemma}[section]
\newtheorem{example}{Example}[section]

\theoremstyle{definition}
\newtheorem{definition}{Definition}[section]

\theoremstyle{remark}

\begin{document}

\title{Higher-order circuits}
\date{}
\author{Matt Wilson}
\email{matthew.wilson@centralesupelec.fr}
\affiliation{Université Paris-Saclay, CentraleSupélec, Inria, CNRS, LMF, 91190 Gif-sur-Yvette, France}

\maketitle

\begin{abstract}
We write down a series of basic laws for (strict) higher-order circuit diagrams. More precisely, we define higher-order circuit theories in terms of: (a) nesting, (b) temporal and spatial composition, and (c) equivalence between lower-order bipartite processes and higher-order bipartite states. 
In category-theoretic terms, these laws are expressed using enrichment and cotensors in symmetric polycategories, along with a frobenius-like coherence between them.
We describe how these laws capture the salient features of higher-order quantum theory, and discover an upper bound for higher-order circuits: any higher-order circuit theory embeds into the theory of strong profunctors.
\end{abstract}

In the past two decades of quantum information theory, and decade of applied category theory, a simple concept has been studied in considerable depth, the notion of a (circuit-theoretic) \textit{hole}. For the physicist, the hole might represent: the most general kind of manipulation which can be applied to a computational gate \cite{Chiribella2008TransformingSupermaps}, the most general possible environment into-which a process could be embeded \cite{Oreshkov2012QuantumOrder}, a node within a causal structure \cite{BarrettCyclicModels}, a generalisation of the notion of non-markovian process to the quantum regime \cite{Pollock2015Non-MarkovianCharacterisation}, and a quantum generalisation for interactive games \cite{Gutoski_2007} and agent-policies \cite{wilson2026agentpolicieshigherordercausal}. On the other hand, in applied category theory, holes are modelled typically using coend and profunctor optics \cite{Riley2018CategoriesOptics, Clarke2020ProfunctorUpdate, Roman2020CombFeedback, Boisseau2022CorneringOptics, Roman2020OpenCalculus, earnshaw2023produoidalalgebraprocessdecomposition}, with connections between these alternative viewpoints explored in \cite{Hefford2022CoendCombs, hefford_prof_sup,hefford2025bvcategoryspacetimeinterventions}.

It would not be surprising to find such a concept to be ubiquitous. A function of a function, a morphism of a morphism, is no big stretch of the imagination, particularly in the light of the lambda calculus and its proposed quantum generalisations \cite{Selinger_2005}. Despite this, an open problem remains in our understanding of the concept of circuit-theoretic holes - the completion of a somewhat deceivingly simple to state analogy:
\begin{align}
\textrm{Circuits} \quad &  \leftrightarrow \quad  \textrm{Monoidal Categories}, \\
(\textrm{Higher-order}) \textrm{ Circuits} \quad &  \leftrightarrow \quad  (\textrm{Higher-order}) \textrm{ Monoidal Categories}.
\end{align}
In other words, we do not have a definition of what it means according to category theory to be higher-order over a monoidal category.

A natural proposal for higher-order monoidal categories, is to use the already established categorical notion of \textit{closed} monoidal categories \cite{Eilenberg1966ClosedCategories}. Closure in a monoidal category captures a salient feature of higher-order functions, namely, that they can be curried to give a natural isomorphism of the form $\mathcal{C}(a \otimes b, c) \cong \mathcal{C}(a, [b,c])$. This is a structure with a notion of space, and a notion of higher order morphisms as those of type $[a,a'] \rightarrow [b,b']$ in $\cat{C}$. The closed monoidal approach was reinterpreted as a foundational base upon which categories of holes could be studied (specifically in the context of quantum causality) in \cite{Wilson2021CausalityTheories}, however, this view adds something unnecessary and misses something necessary to the study of circuit-theoretic holes:
\begin{itemize}
\item The existence of iterated higher-order structure in closed monoidal categories requires iterated notions of higher-order morphism, which appear unnecessary for the qualification of a theory as a theory of holes. Indeed, in the theory of quantum supermaps, the majority of attention is given to those higher-order maps which are no-more than $2^{nd}$-order \cite{Chiribella2008TransformingSupermaps, Oreshkov2012QuantumOrder}. 
\item A key aspect of holes as modelled by quantum supermaps is that one need not insert an entire process into a hole, instead one might only insert part of a multipartite process into a hole. There is nothing in the structure of a closed monoidal category which allows for this. One cannot for instance ask that with respect to the closed monoidal product  \[ \mathcal{C}(A , A') \otimes \mathcal{C}(B,B')  \cong  \mathcal{C}(A  \otimes  B,A'  \otimes  B').  \] Indeed, to ask for such an isomorphism as argued in \cite{wilson_polycategories}, is close to request the existence of pathological time loops\footnote{In-fact, for higher order circuit theories as they will be defined in this paper, it is conjectured by the author that the existence of such as isomorphism is in-fact precisely the requirement that the category at hand be a traced symmetric monoidal category. A proof of this fact is left for future work.}. 
\end{itemize}

The first point can be resolved by replacing the notion of a \textit{closed} symmetric monoidal category with the notion of a $\cat{V}$-enriched symmetric monoidal category $\cat{C}$ for some $\cat{V}$, referred to from now on simply as a $\cat{V}$-smc $\cat{C}$ for short \cite{Johnstone1983BASIC64}. This view on higher-order processes explored in \cite{Rennela2018ClassicalTheory} and \cite{wilson2023mathematical}, allows for higher-order without requiring the existence of higher-higher-order maps and their iterations. Holes for $\cat{C}$ are taken to exist in $\cat{V}$, but no claim is made regarding holes for the higher-order theory $\cat{V}$. The fact that this view gives a non-iterated generalisation of closed monoidal categories is further supported by the observation that mergers of infinite sequences of such enriched monoidal categories, do in-fact form categories which are closed \cite{wilson2023mathematical}.

Regarding the second point, a seminal paper on the compositional properties of higher-order quantum operations \cite{kissinger_caus}, identifies that they form not just closed monoidal categories but in fact star-autonomous categories. As a result, quantum higher-order operations can be equipped with a pair of linearly distributive tensor products \cite{COCKETT1997133}, in which an additional non-closed tensor called the par correctly builds $\mathcal{C}(A  \otimes  B,A'  \otimes  B')$ from $ \mathcal{C}(A , A')$ and $ \mathcal{C}(B , B')$. The $*$-autonomous viewpoint does again, however, assume the existence of higher-higher-iterated objects. 

As argued in \cite{wilson_polycategories}, the second point can be resolved in a theory-independent setting by relaxing the possibility that $\cat{V}$ itself be a symmetric monoidal category, requiring instead only that it might form a symmetric polycategory \cite{Szabo01011975}. Symmetric polycategories are both perfectly reasonable bases for enrichment \cite{LEINSTER2002391}, and what remains when representability is stripped away from categories which are linearly distributive. What polycategorical structure gives in this setting, is precisely the possibility of having partial applicability of holes onto all bipartite processes without a heuristic production of time loops. 
In short then, a minimal resolution to the above points appears to be that an axiomatisation for higher-order circuits should not be based on closed symmetric monoidal categories, but instead on $\cat{V}$-symmetric monoidal categories for some symmetric polycategory $\cat{V}$. 

This proposal to study higher-order circuits as monoidal categories enriched in symmetric polycategories is quite distinct from the traditional closed-monoidal view on what it means to be higher-order. It is certainly broad enough that it avoids imposing unnecessary structure onto theories of higher-order maps, however it is natural to ask whether it might be still too broad to capture everything that is necessary about higher-order processes.
In this article, we therefore consider the direct consequences of the enriched and polycategorical view on higher-order circuits, giving evidence of their suitability. 
We add to the enriched polycategorical view, a few additional morphisms such as cotensors \cite{COCKETT1997133}, which come with a reasonable interpretation in terms of wires into-which gaps have been inserted (in analogy with the identity morphisms of categories, which simply represent wires without gaps) \cite{coecke_kissinger_2017, joyal_street1}. We then consider the various compatibility laws between the morphisms induced by both enrichment and the existence of cotensors, which represent the non-ambiguity in algebraic interpretation of pictures which convey complex arrangements of wires into which multiple gaps have been cut. 

As formal evidence that higher-order circuits as defined here are not too broad, we outline what we interpret as an upper-bound theorem for higher-order circuits. First, for an arbitrary symmetric monoidal category $\cat{C}$ and higher-order circuit theory $\cat{P}$ over $\cat{C}$, we note that one can always construct it's operational closure $\cat{P}^{\#}$ which simply quotients holes by their action on arbitrary states. Then, we show that there exists a multifunctor $\cat{P}^{\#}  \rightarrow \mathbf{StProf}[\cat{C}]$ into the multicategory of Tambara modules (equivalently Strong Profunctors) \cite{PARE199833, Tambara2006, pastro2007doublesmonoidalcategories} induced by the day tensor product \cite{Day_1977}. 

On the one-hand, this gives a formal argument for the strong-profunctorial \cite{wilson_polycategories, wilson_locality, hefford_supermaps,hefford2025bvcategoryspacetimeinterventions} approach to the study of holes, by demonstrating that the operational behaviour of any hole in any higher-order circuit theory, can be interpreted as a strong natural transformation between strong profunctors. Conversely, the characterisation theorem for locally-applicable transformations (which are, in-fact, strong natural transformations between strong profunctors \cite{hefford_supermaps}) on quantum theory \cite{wilson2025quantumsupermapscharacterizedlocality}, along with this upper-bound theorem demonstrates that on finite-dimensional quantum theory there can be no higher-order circuit theory which supports anything beyond the theory of quantum supermaps \cite{Chiribella2013QuantumStructure, Chiribella2008QuantumArchitectureb, Chiribella2008TransformingSupermaps} - evidencing the claim that higher-order circuits as defined here are indeed not too broad.

Within this article many proofs and definitions are presented at the level of sketches, with for instance the introduction of graphical notations which have not been formally proven sound. Statements are limited to the strict setting \cite{maclane1998categories}, with a key purpose of this article being to motivate full categorification of these ideas, without allowing for the desire for perfect formality or full categorification to hinder development of applications in quantum information theory and quantum foundations in which the level of formality of this article is already beyond what is standard.
Applications within categorical quantum mechanics \cite{Abramsky2004AProtocols} and the foundations of quantum information theory are also reflected on. In particular, the abstract framework presented here gives a new way to define resource theories of higher-order processes \cite{Coecke2014AResources, GourDynamicalResources}, and might contribute towards the clean axiomatisation of higher-order processes \cite{Giacomini_2016, Chiribella2013NormalCP, wilson_polycategories, wilson_locality, hefford_supermaps,hefford2025bvcategoryspacetimeinterventions, bavaresco2024indefinitecausalorderboxworld, sengupta2024achievingmaximalcausalindefiniteness} within post-quantum \cite{barrett_gpts} and non-finite-dimensional theories.

\section{Axiomatising higher-order circuits}

In this section we will introduce a series of axioms for higher-order circuits. We begin with the most fundamental structure; enrichment and multi-arity composition. Following this, we will introduce the important additional feature of cotensors and the laws they induce.
\subsection{Enrichment and Multi-Arity Composition Laws}
Let us begin by recalling the definition of a strict $\mathcal{P}$-smc with $\mathcal{P}$ a symmetric polycategory. Such polycategories will form the backbone of higher-order circuit theories. 

\begin{definition}
A strict $\mathcal{P}$-smc with $\mathcal{P}$ a symmetric polycategory is
\begin{itemize}
\item A collection of objects $a,b,c, \dots$
\item For each pair of objects $a,b$ an associated object in $\mathcal{P}$ denoted $[a,b]$
\item For each triple of objects $a,b,c$ an assocaited \textit{sequential composition} morphism \[ \tikzfigscale{0.8}{figs/hole_circ} \rightarrow [a,c]       \quad : \quad [a,b] [b,c]   \]
\item For each tuple of objects $a,a',b,b'$ an associated \textit{paralell composition} morphism \[  \tikzfigscale{0.8}{figs/hole_24b}   \quad : \quad [a,a'][b,b'] \rightarrow [ab,a'b']   \]
\item For each object $a$ an associated \textit{identity} morphism\[ \tikzfigscale{0.8}{figs/hole_identity_formal}  \quad : \quad i: \emptyset \rightarrow [a,a]\]
\end{itemize}
these structural morphisms are furthermore required to satisfy a variety of laws, which directly mimic the laws of strict symmetric monoidal categories. Namely \textit{sequential associativity}, \textit{tensor associativity}, \textit{interchange}, \textit{sequential unitality}, and \textit{tensor unitality}\footnote{We will give an explicit definition and interpretation of each of these laws later in the text.}. 
\end{definition}

Let us now give an outline following \cite{wilson2023compositional}, on the interpretation of such polycategories as bare-minimum requirements for theories of circuit-theoretic holes. 

\subsubsection{Polycategorical Composition of Holes}

The basic atom of a category is a wire - for us the basic atom is a hole, that is, a wire with a gap depicted as such \[  \tikzfigscale{0.8}{figs/hole_1}  \quad : \quad  \tikzfigscale{0.8}{figs/hole_1b}.  \]
Here we have given the intuitive picture on the left and it's associated formal interpretation into types and morphisms on the right-hand side. 
Operations on holes are depicted as follows: \[    \tikzfigscale{0.8}{figs/hole_2} \quad : \quad  \tikzfigscale{0.8}{figs/hole_3t} ,  \]
where clearly we should be able to compose such operations, and such a composition should come equipped with an identity operation  \[    \tikzfigscale{0.8}{figs/hole_5} \quad : \quad  \tikzfigscale{0.8}{figs/hole_6t}   \quad \quad \quad \quad   \tikzfigscale{0.8}{figs/hole_7} \quad : \quad  \tikzfigscale{0.8}{figs/hole_8t}.   \]
Just as in traditional circuits there is an identity operation on wires - depicted with the wire, in higher-order circuits there is an identity operation on holes - depicted by the hole.   
We can extend to the multi-hole picture in-which multiple-input holes and multiple-output holes are accounted for as follows \[    \tikzfigscale{0.8}{figs/hole_9} \quad : \quad  \tikzfigscale{0.8}{figs/hole_10t}   \quad \quad \quad \quad   \tikzfigscale{0.8}{figs/hole_11} \quad : \quad  \tikzfigscale{0.8}{figs/hole_12t} .  \]
Generally, we can imagine $(n \rightarrow m)$-arity maps on holes for any integers $n,m$  \[    \tikzfigscale{0.8}{figs/hole_13} \quad : \quad  \tikzfigscale{0.8}{figs/hole_14t} .  \]

Multiple-outputs allow us to introduce partial-nesting, this is \textit{the} key idea in both the linear-algebraic \cite{chiribella_supermaps} and categorical \cite{hefford_supermaps} approaches to defining quantum supermaps  \[    \tikzfigscale{0.8}{figs/hole_15} \quad : \quad  \tikzfigscale{0.8}{figs/hole_planar_1} .  \]

\subsubsection{Enriched Structure of Holes}
Given any two boxes, we should be able to put them in sequence or in parallel to interpret them as a new $2$-input hole
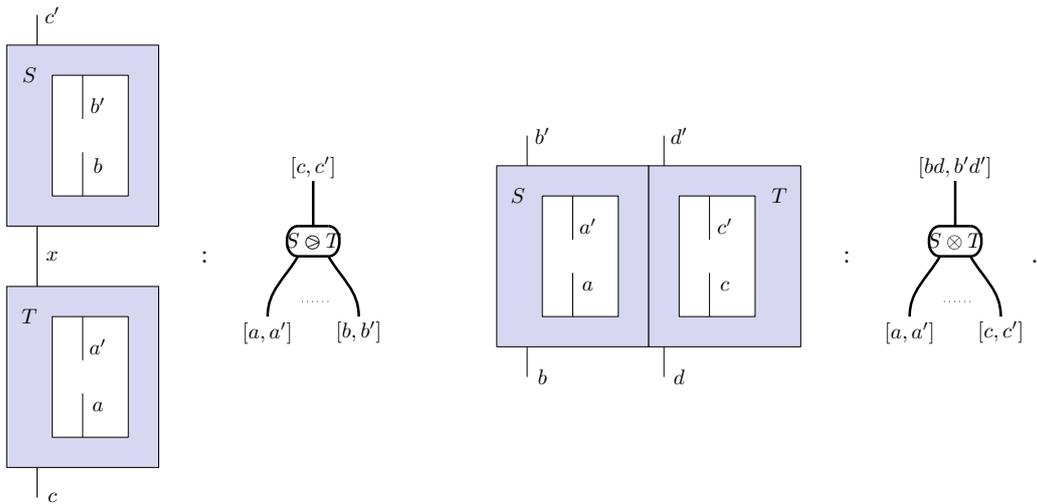
\[    \tikzfigscale{0.8}{figs/hole_17} \quad : \quad  \tikzfigscale{0.8}{figs/hole_18c}  \quad \quad \quad \quad   \tikzfigscale{0.8}{figs/hole_21} \quad : \quad  \tikzfigscale{0.8}{figs/hole_22}. \] In-fact, this should not stop at pairs of holes. For instance, we should be able to form terms such as \[    \tikzfigscale{0.8}{figs/hole_multi_seq_1} \quad : \quad  \tikzfigscale{0.8}{figs/hole_18d} \] and \[    \tikzfigscale{0.8}{figs/hole_multi_par_1} \quad : \quad  \tikzfigscale{0.8}{figs/hole_18e}. \] 
Rather than introducing entire novel composition rules to represent these terms, one can simply construct them from the existence of additional structural morphisms. Just as any number of wires defines a canonical map on wires, any number of gaps in wires defines a canonical map on gaps in wires. Consequently, if we introduce the following canonical maps \[    \tikzfigscale{0.8}{figs/hole_19b} \quad : \quad  \tikzfigscale{0.8}{figs/hole_circ}  \quad \quad \quad \quad  \tikzfigscale{0.8}{figs/hole_23b} \quad : \quad  \tikzfigscale{0.8}{figs/hole_24b}, \]   we can then define general composition of boxes through 
 \[      \tikzfigscale{0.8}{figs/hole_18c}     \quad = \quad    \tikzfigscale{0.8}{figs/hole_seq_1} \quad  \quad \quad  \quad         \tikzfigscale{0.8}{figs/hole_22}     \quad = \quad    \tikzfigscale{0.8}{figs/hole_par_1}    .      \]
 Finally, the canonical way to fill a gap in a wire, is with a wire. In other words, the canonical state in higher-order circuits is the identity  \[   \tikzfigscale{0.8}{figs/hole_identityb} \quad = \quad  \tikzfigscale{0.8}{figs/hole_identity_formal} .  \] 
 
 \subsubsection{Enriched Laws for Holes}
We can now impose \textit{sequential associativity} and \textit{tensor associativity} of $\bigcircle$ and $\bigotimes$ respectively with the following equations \[   \tikzfigscale{0.8}{figs/hole_ass_1} \quad = \quad  \tikzfigscale{0.8}{figs/hole_ass_2}  \quad \quad \quad \quad  \tikzfigscale{0.8}{figs/hole_ass_3} \quad = \quad  \tikzfigscale{0.8}{figs/hole_ass_4} . \]
These equations model basic graphical tautologies at the level of the holes they are intended to represent, the tautology for sequential holes being
  \[   \tikzfigscale{0.8}{figs/diagram_hole_circ_1b} \quad : \quad  \tikzfigscale{0.8}{figs/diagram_hole_circ_2b} ,  \]
and the tautology for parallel holes being \[   \tikzfigscale{0.8}{figs/diagram_hole_par_1b} \quad : \quad  \tikzfigscale{0.8}{figs/diagram_hole_par_2b} .  \]

Associativity at the level of general box composition will come for free by unpacking as follows  \[   \tikzfigscale{0.8}{figs/hole_seq_ass_1} \quad = \quad  \tikzfigscale{0.8}{figs/hole_seq_ass_2}  \quad  = \quad  \tikzfigscale{0.8}{figs/hole_seq_ass_3}.  \]

The \textit{interchange} law is imposed by the following equation
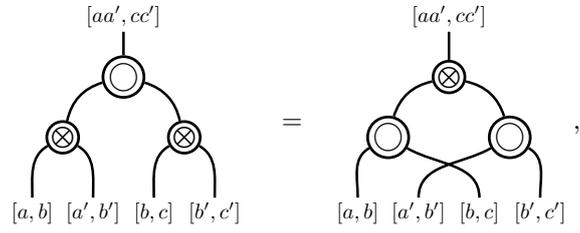
\[   \tikzfigscale{0.8}{figs/formal_interchange_1} \quad = \quad  \tikzfigscale{0.8}{figs/formal_interchange_2} ,  \] which in turn captures a more non-trivial four-hole tautology
 \[   \tikzfigscale{0.8}{figs/hole_interchange_1bb} \quad = \quad  \tikzfigscale{0.8}{figs/hole_interchange_2b} .  \] 
The \textit{sequential unit} law states that  \[   \tikzfigscale{0.8}{figs/hole_unit_law_1} \quad = \quad  \tikzfigscale{0.8}{figs/hole_unit_law_2} \quad = \quad  \tikzfigscale{0.8}{figs/hole_unit_law_3}   \] which has a graphical interpretation as 
\[   \tikzfigscale{0.8}{figs/hole_identity_pic_1} \quad = \quad  \tikzfigscale{0.8}{figs/hole_identity_pic_2} .  \]  
Finally, the \textit{tensor unit} law states that
  \[  \tikzfigscale{0.8}{figs/hole_tensor_unit_formal_1}  \quad  = \quad \tikzfigscale{0.8}{figs/hole_unit_law_2},   \] 
  and has the following graphical interpretation
   \[  \tikzfigscale{0.8}{figs/hole_tensor_unit_1}  \quad  = \quad  \tikzfigscale{0.8}{figs/hole_tensor_unit_2}.   \]



\subsection{Higher-order circuits}
We now give a full formal definition of higher-order cirucit theories, which are enrichements into polycategories as outlined above, but equipped with cotensors which model the existence of parallel gaps in parallel pairs of wires. 
\begin{definition}[Higher-order circuit theory]
A Higher-order circuit theory is a strict $\mathcal{P}$-monoidal category with $\mathcal{P}$ a symmetric polycategory equipped with a cotensor $\bullet : [a_1 \dots a_n \rightarrow b_1 \dots b_n] \rightarrow [a_1 ,b_1] \dots [a_n, b_n]$ satisfying 
\begin{itemize}
\item Associativity laws  \[   \tikzfigscale{0.8}{figs/theory_def_4}  \quad = \quad  \tikzfigscale{0.8}{figs/theory_def_5}    \]
\item Frobenius laws  \[   \tikzfigscale{0.8}{figs/hole_frobenius_1} \quad = \quad  \tikzfigscale{0.8}{figs/hole_frobenius_2} \quad \quad \quad \quad     \tikzfigscale{0.8}{figs/hole_frobenius_3} \quad = \quad  \tikzfigscale{0.8}{figs/hole_frobenius_4}   \]
\item The copy law \[   \tikzfigscale{0.8}{figs/hole_cotensor_copy_1} \quad = \quad  \tikzfigscale{0.8}{figs/hole_cotensor_copy_2} .  \]
\end{itemize}
A higher-order circuit theory is referred to as braided if $\cat{C}$ is braided (enriched again in $\mathcal{P}$) and furthermore the cotensor satsifes the following braid law   \[   \tikzfigscale{0.8}{figs/bold_symmetry_1}  \quad = \quad  \tikzfigscale{0.8}{figs/bold_symmetry_2},    \]   where we use $\beta_{ab}: a b \rightarrow ba$ to refer to the braid of $\cat{C}$. 
\end{definition}
Let us now unpack the additional features outlined above as a requirement for higher-order circuit theories. 
\subsubsection{Cotensors} Formally, a cotensor for a pair of objects $b_1, b_2$  in a polycategory is an object $ b_1 \boxtimes b_2$ and morphism $\bullet: b_1 \boxtimes b_2 \rightarrow b_1 b_2$ such that \[  \forall \ S:\underline{x} \rightarrow b_1 b_2 \ \exists  ! \ S': \underline{x} \rightarrow b_1 \boxtimes b_2   \ s.t. \   \tikzfigscale{0.8}{figs/hole_cotensor_universal_2} \ = \   \tikzfigscale{0.8}{figs/hole_cotensor_universal_3} .  \]

For higher-order circuits, we require the existence of a cotensor which inherits it's action on objects by the action of the monoidal product of $\cat{C}$ on the underlying objects of $\cat{C}$. More precisely, we recall that we have objects of the form $[a,b]$ and ask for the cotensor to evaluate as $[a_1 , b_1] \boxtimes \dots \boxtimes[a_n , b_n] = [a_1 \dots a_n, b_1 \dots b_n]$.

This cotensor requirement comes naturally from noting that within the definition of a box with holes there is an ambiguity in the interrpetation of output wires. One could for instance have equally interpreted the diagram on the left as either of the formal diagrams on the right
 \[    \tikzfigscale{0.8}{figs/hole_13} \quad : \quad  \tikzfigscale{0.8}{figs/hole_14t}  \quad \quad \text{or} \quad \quad    \tikzfigscale{0.8}{figs/hole_14tb}.  \]
This tells us that we should see an isomorphism of the kind \[\mathcal{P}(\underline{x} ,\underline{y} [a_1 \dots a_n , b_1 \dots b_n]) \cong \mathcal{P}(\underline{x} , \underline{y} [a_1 , b_1] \dots [ a_n , b_n]),\] suitably natural with respect to compositions via polymorphisms on $\underline{x}$ and $\underline{y}$. 
By the Yoneda lemma for polycategories \cite{Shulman_2021}, one can see that the existence of such isomorphisms would be captured by a particular canonical box formalised by the cotensor. Indeed, we use the cotensor to represent the following wire-gap hole  \[    \tikzfigscale{0.8}{figs/hole_cotensor} \quad : \quad  \tikzfigscale{0.8}{figs/hole_cotensor_1} .  \] In short, the cotensor looks graphically as-if it is the identity, which is fitting since it exists only to identify between lists of objects and composite objects. 


\subsubsection{Coherence laws between cotensors and enrichment}
The \textit{frobenius laws} given by
 \[   \tikzfigscale{0.8}{figs/hole_frobenius_1} \quad = \quad  \tikzfigscale{0.8}{figs/hole_frobenius_2} \quad \quad \quad \quad     \tikzfigscale{0.8}{figs/hole_frobenius_3} \quad = \quad  \tikzfigscale{0.8}{figs/hole_frobenius_4} ,  \] model the following relationships between ways of arranging and composing wire-gaps \[  \tikzfigscale{0.8}{figs/hole_frob_diagram_1} \quad = \quad  \tikzfigscale{0.8}{figs/hole_frob_diagram_2}   \] 
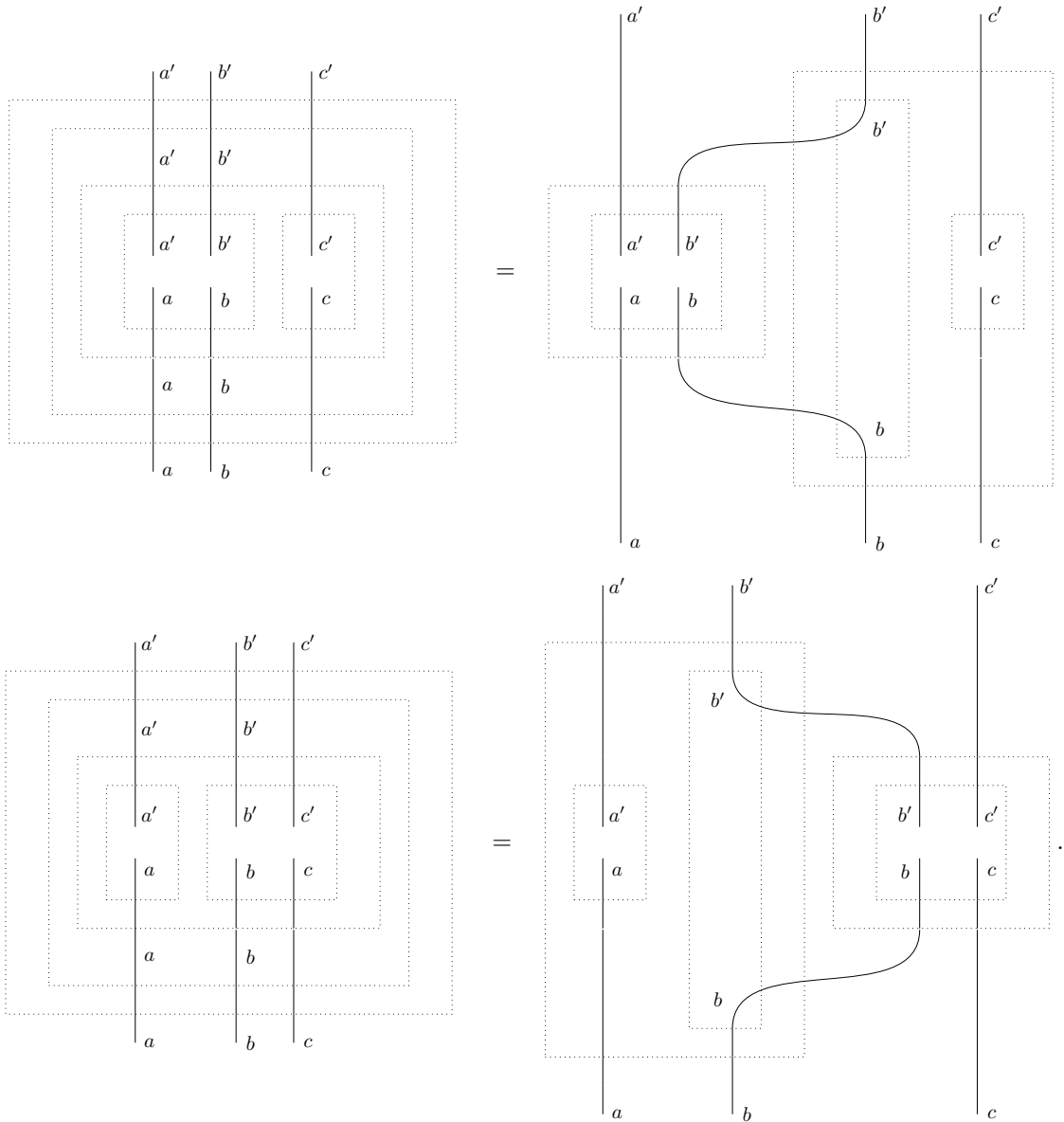
 \[  \tikzfigscale{0.8}{figs/hole_frob_diagram_3} \quad = \quad  \tikzfigscale{0.8}{figs/hole_frob_diagram_4} .  \] 
These laws say that for all intents and purposes, the cotensor does nothing except to let us examine subsystems. 

Next, the copying of the sequential unit through the cotensor  \[   \tikzfigscale{0.8}{figs/hole_cotensor_copy_1} \quad = \quad  \tikzfigscale{0.8}{figs/hole_cotensor_copy_2} ,  \] has a graphical interpretation in terms of tautology for parallel wires \[   \tikzfigscale{0.8}{figs/hole_cotensor_copy_d1} \quad = \quad  \tikzfigscale{0.8}{figs/hole_cotensor_copy_d2}  . \] 
Note finally that the braid law   \[   \tikzfigscale{0.8}{figs/bold_symmetry_1}  \quad = \quad  \tikzfigscale{0.8}{figs/bold_symmetry_2},     \]  encodes unambiguos interpretation of the following diagram \[  \tikzfigscale{0.8}{figs/two_braids}. \]
As we show in the appendix, the braid law actually implies the following more elaborate result for braids by induction  \[   \tikzfigscale{0.8}{figs/bold_braid_1}  \quad = \quad  \tikzfigscale{0.8}{figs/bold_braid_2}.   \]  
This completes the intuitive graphical interpretation of the axioms of higher-order circuit theories.
\subsubsection{Examples of higher-order circuit theories}
Let us finally see a few important examples.
\begin{example}
Every compact closed category \cite{KELLY1980193} defines a higher-order circuit theory over itself. Enrichment is inherited directly from closure so that $[a,a']= a* \otimes a`$, polycategorical structure is given by simply forgetting the monoidal product (passing through the forgetful functor from linearly distributive categories to polycategories, in the special case of a linearly distributive category in which the tensor and par coincide \cite{COCKETT1997133}). The cotensor is given by the braid. The associativity laws, frobenius laws, and the copy law are then all trivially satisfied by the axioms of compact closed categories.
\end{example}

\begin{example}
Let $\mathbf{C}$ be a precausal category, and recall that $\mathbf{Caus}[\mathbf{C}]$ is a star-autonomous, and so linearly-distributive, category \cite{kissinger_caus} with a tensor $\otimes$ and a par $\boxtimes$. Consider the symmetric monoidal category $[\mathbf{C}]_1$ which arises from the full subcategory of $\mathbf{Caus}[\mathbf{C}]$ generated by first order objects $a,b, \dots$ and the tensor (which for first order objects coincides with the par). Consider also the full symmetric linearly distributive subcategory $[\mathcal{C}]_2$ generated by all objects of the form $[a,a']$ under both the tensor and the par, with each of $a,a'$ given by first-order objects.

By the $*$-autonomous structure of $\mathbf{Caus}[\mathcal{C}]$ we see that $[\mathcal{C}]_1$ is symmetric monoidal enriched in $[\mathcal{C}]_2$. Furthermore, the forgetful functor from linearly distributive categories to symmetric polycategories \cite{COCKETT1997133} can be applied to give a (with minor abuse of notation) $[\mathcal{C}]_2$-enriched smc $[\mathcal{C}]_1$ with $[\mathcal{C}]_2$ a symmetric polycategory. 

It is shown in \cite{kissinger_caus} that there exists an isomorphism in $\mathbf{Caus}[\mathcal{C}]_2$ of a form $[a \otimes b , a'\otimes b'] \cong [a,a'] \boxtimes [b,b']$. Under the forgetful functor this isomorphism becomes the cotensor. The frobenius and copy laws, are then all inherited directly from the compact closed structure of $\mathcal{C}$ as in the previous example.
\end{example}

\begin{example}
 Consider any strict symmetric monoidal subcategory $\mathcal{C} \subseteq \mathcal{D}$ of a strict compact closed category $\mathcal{D}$. A $\mathcal{D}$-supermap on $\mathcal{C}$ \cite{wilson2025quantumsupermapscharacterizedlocality} of type \[ \otimes_{i}(a_i \Rightarrow a_i') \rightarrow (b \Rightarrow b') \]  is a process of type $S: b \otimes (a_1' \otimes \dots \otimes a_n') \rightarrow b' \otimes (a_1 \otimes \dots \otimes a_n)$ in $\cat{D}$ such that for every collection $\phi_i : a_i \otimes x_i \rightarrow a_i' \otimes x_i'$ of processes in $\cat{C}$ the following process is also in $\cat{C}$. 
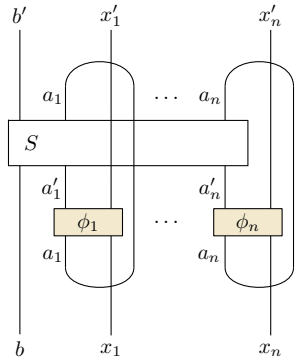
\[ \begin{tikzpicture}[scale=0.6, {every node/.style}={scale=0.8}]
	\begin{pgfonlayer}{nodelayer}
		\node [style=none] (117) at (-1.25, -3.55) {$b$};
		\node [style=none] (118) at (-1.25, 3.825) {$b'$};
		\node [style=none] (123) at (-1.5, 1.5) {};
		\node [style=none] (124) at (-1.5, 0.5) {};
		\node [style=none] (125) at (3.75, 1.5) {};
		\node [style=none] (126) at (3.75, 0.5) {};
		\node [style=none] (127) at (-1.25, 3.5) {};
		\node [style=none] (128) at (-1.25, -3.225) {};
		\node [style=none] (129) at (-1.25, 1.5) {};
		\node [style=none] (130) at (-1.25, 0.5) {};
		\node [style=none] (131) at (-1, 1) {$S$};
		\node [style=none] (134) at (-0.25, 2.25) {};
		\node [style=none] (135) at (-0.25, 1.5) {};
		\node [style=none] (136) at (-0.25, -0.45) {};
		\node [style=none] (137) at (-0.25, 0.5) {};
		\node [style=none] (196) at (-0.5, -0.45) {};
		\node [style=none] (197) at (-0.5, -1.05) {};
		\node [style=none] (198) at (1, -0.45) {};
		\node [style=none] (199) at (1, -1.05) {};
		\node [style=none] (200) at (-0.25, -0.45) {};
		\node [style=none] (202) at (0.25, -0.75) {$\phi_1$};
		\node [style=none] (206) at (0.25, -0.425) {};
		\node [style=none] (207) at (1.25, 2.25) {};
		\node [style=none] (208) at (1.25, -1.75) {};
		\node [style=none] (209) at (-0.25, -1.75) {};
		\node [style=none] (210) at (-0.25, -1.05) {};
		\node [style=none] (211) at (0.75, -1.05) {};
		\node [style=none] (212) at (0.75, -0.45) {};
		\node [style=none] (213) at (0.75, 3.5) {};
		\node [style=none] (214) at (0.75, -3.25) {};
		\node [style=none] (215) at (0.75, -3.575) {$x_1$};
		\node [style=none] (216) at (0.75, 3.825) {$x_1'$};
		\node [style=none] (217) at (-0.525, -1.5) {$a_1$};
		\node [style=none] (218) at (-0.55, 0) {$a_1'$};
		\node [style=none] (219) at (-0.525, 2) {$a_1$};
		\node [style=none] (300) at (3.25, 2.25) {};
		\node [style=none] (301) at (3.25, 1.5) {};
		\node [style=none] (302) at (3.25, -0.45) {};
		\node [style=none] (303) at (3.25, 0.5) {};
		\node [style=none] (304) at (3, -0.45) {};
		\node [style=none] (305) at (3, -1.05) {};
		\node [style=none] (306) at (4.5, -0.45) {};
		\node [style=none] (307) at (4.5, -1.05) {};
		\node [style=none] (308) at (3.25, -0.45) {};
		\node [style=none] (309) at (3.75, -0.75) {$\phi_n$};
		\node [style=none] (310) at (3.75, -0.425) {};
		\node [style=none] (311) at (4.75, 2.25) {};
		\node [style=none] (312) at (4.75, -1.75) {};
		\node [style=none] (313) at (3.25, -1.75) {};
		\node [style=none] (314) at (3.25, -1.05) {};
		\node [style=none] (315) at (4.25, -1.05) {};
		\node [style=none] (316) at (4.25, -0.45) {};
		\node [style=none] (317) at (4.25, 3.5) {};
		\node [style=none] (318) at (4.25, -3.25) {};
		\node [style=none] (319) at (4.25, -3.575) {$x_n$};
		\node [style=none] (320) at (4.25, 3.825) {$x_n'$};
		\node [style=none] (321) at (2.9, -1.5) {$a_n$};
		\node [style=none] (322) at (2.925, 0) {$a_n'$};
		\node [style=none] (323) at (2.925, 2) {$a_n$};
		\node [style=none] (324) at (2, 2) {$\dots$};
		\node [style=none] (325) at (2, -0.75) {$\dots$};
		\node [style=none] (326) at (0.75, 1.75) {};
		\node [style=none] (327) at (1.25, 0.25) {};
		\node [style=none] (328) at (0.75, 0.25) {};
		\node [style=none] (329) at (1.25, 1.75) {};
	\end{pgfonlayer}
	\begin{pgfonlayer}{edgelayer}
		\draw (123.center)
			 to (125.center)
			 to (126.center)
			 to (124.center)
			 to cycle;
		\draw (129.center) to (127.center);
		\draw (130.center) to (128.center);
		\draw [in=90, out=-90] (134.center) to (135.center);
		\draw [in=90, out=-90] (137.center) to (136.center);
		\draw [style={morphism_shade}] (196.center)
			 to (198.center)
			 to (199.center)
			 to (197.center)
			 to cycle;
		\draw (210.center) to (209.center);
		\draw (211.center) to (214.center);
		\draw [bend left=90, looseness=1.25] (134.center) to (207.center);
		\draw [bend left=90] (208.center) to (209.center);
		\draw [in=90, out=-90] (300.center) to (301.center);
		\draw [in=90, out=-90] (303.center) to (302.center);
		\draw [style={morphism_shade}] (304.center)
			 to (306.center)
			 to (307.center)
			 to (305.center)
			 to cycle;
		\draw (314.center) to (313.center);
		\draw (317.center) to (316.center);
		\draw (315.center) to (318.center);
		\draw [bend left=90, looseness=1.25] (300.center) to (311.center);
		\draw (311.center) to (312.center);
		\draw [bend left=90] (312.center) to (313.center);
		\draw (213.center) to (326.center);
		\draw (328.center) to (212.center);
		\draw (326.center) to (328.center);
		\draw (329.center) to (327.center);
		\draw (327.center) to (208.center);
		\draw (207.center) to (329.center);
	\end{pgfonlayer}
\end{tikzpicture} \] 
 
 We can define a symmetric polycategory by defining the polymorphisms of type $\star_i [a_i,a_i'] \rightarrow \star_k[b_k,b_k']$ to be the $\mathcal{D}$-supermaps on $\mathcal{C}$ of type $\star_i [a_i,a_i'] \rightarrow [\otimes_k b_k, \otimes_k b_k']$. It is proven that this is a symmetric polycategory in \cite{wilson_polycategories}, and it is easy to see furthermore that $\mathcal{C}$ is symmetric monoidal enriched in this polycategory $[\cat{D}]\mathbf{sup}[\cat{C}]$. The cotensor is defined by 
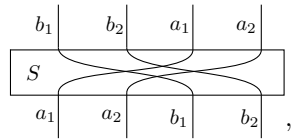
\[ \begin{tikzpicture}[scale=0.6, {every node/.style}={scale=0.8}]
	\begin{pgfonlayer}{nodelayer}
		\node [style=none] (123) at (-2.25, 1.5) {};
		\node [style=none] (124) at (-2.25, 0.5) {};
		\node [style=none] (125) at (3.75, 1.5) {};
		\node [style=none] (126) at (3.75, 0.5) {};
		\node [style=none] (131) at (-1.75, 1) {$S$};
		\node [style=none] (134) at (1.75, 2.5) {};
		\node [style=none] (135) at (1.75, 1.5) {};
		\node [style=none] (136) at (1.75, -0.45) {};
		\node [style=none] (137) at (1.75, 0.5) {};
		\node [style=none] (200) at (1.75, -0.45) {};
		\node [style=none] (218) at (1.45, 0) {$b_1$};
		\node [style=none] (219) at (1.475, 2) {$a_1$};
		\node [style=none] (300) at (3.25, 2.5) {};
		\node [style=none] (301) at (3.25, 1.5) {};
		\node [style=none] (302) at (3.25, -0.45) {};
		\node [style=none] (303) at (3.25, 0.5) {};
		\node [style=none] (308) at (3.25, -0.45) {};
		\node [style=none] (323) at (2.925, 2) {$a_2$};
		\node [style=none] (326) at (3, 0) {$b_2$};
		\node [style=none] (327) at (-1.2, -0.45) {};
		\node [style=none] (328) at (-1.2, 0.55) {};
		\node [style=none] (329) at (-1.2, 2.5) {};
		\node [style=none] (330) at (-1.2, 1.55) {};
		\node [style=none] (331) at (-1.2, 2.5) {};
		\node [style=none] (332) at (-1.5, 2.05) {$b_1$};
		\node [style=none] (333) at (-1.475, 0.05) {$a_1$};
		\node [style=none] (334) at (0.3, -0.45) {};
		\node [style=none] (335) at (0.3, 0.55) {};
		\node [style=none] (336) at (0.3, 2.5) {};
		\node [style=none] (337) at (0.3, 1.55) {};
		\node [style=none] (338) at (0.3, 2.5) {};
		\node [style=none] (339) at (-0.025, 0.05) {$a_2$};
		\node [style=none] (340) at (0.05, 2.05) {$b_2$};
	\end{pgfonlayer}
	\begin{pgfonlayer}{edgelayer}
		\draw (123.center)
			 to (125.center)
			 to (126.center)
			 to (124.center)
			 to cycle;
		\draw (134.center) to (135.center);
		\draw [in=90, out=-90] (137.center) to (136.center);
		\draw [in=90, out=-90] (300.center) to (301.center);
		\draw [in=90, out=-90] (303.center) to (302.center);
		\draw [in=-90, out=90] (327.center) to (328.center);
		\draw [in=-90, out=90] (330.center) to (329.center);
		\draw [in=-90, out=90] (334.center) to (335.center);
		\draw [in=-90, out=90] (337.center) to (336.center);
		\draw [in=90, out=-90, looseness=0.50] (330.center) to (137.center);
		\draw [in=90, out=-90, looseness=0.50] (337.center) to (303.center);
		\draw [in=90, out=-90, looseness=0.50] (301.center) to (335.center);
		\draw [in=-90, out=90, looseness=0.50] (328.center) to (135.center);
	\end{pgfonlayer}
\end{tikzpicture}, \]  and again the associativity, frobenius, and copy laws all follow from the defining equations of compact closure for $\cat{D}$. 
\end{example}

For the next class of examples we will need to introduce locally-applicable transformations.

\begin{definition}[Locally-applicable tranformations]
A locally-applicable transformation of type \[S: (a \Rightarrow a') \rightarrow (b \Rightarrow b')\] is a family of functions \[   S_{X,X'} :   \cat{C}(a \otimes x, a' \otimes x') \rightarrow \cat{C}(b \otimes x , b' \otimes x')  \] satisfying  
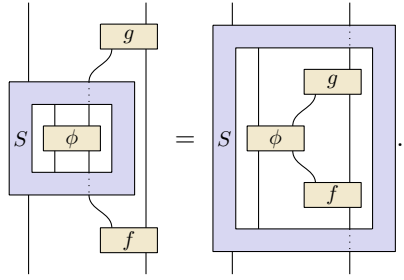
\[   \begin{tikzpicture}[scale=0.6, {every node/.style}={scale=0.8}]
	\begin{pgfonlayer}{nodelayer}
		\node [style=none] (1) at (-1.5, 1.25) {};
		\node [style=none] (2) at (1.25, 1.25) {};
		\node [style=none] (9) at (1.25, -1.25) {};
		\node [style=none] (10) at (-1.5, -1.25) {};
		\node [style=none] (11) at (-1.25, 0) {$S$};
		\node [style=none] (49) at (0.75, 0.75) {};
		\node [style=none] (52) at (-1, 0.75) {};
		\node [style=none] (53) at (-1, -0.75) {};
		\node [style=none] (54) at (0.75, -0.75) {};
		\node [style=none] (184) at (-0.750001, 0.275) {};
		\node [style=none] (185) at (-0.750001, -0.275) {};
		\node [style=none] (186) at (0.499999, 0.275) {};
		\node [style=none] (187) at (0.499999, -0.275) {};
		\node [style=none] (188) at (-0.500001, 0.75) {};
		\node [style=none] (189) at (-0.500001, -0.75) {};
		\node [style=none] (190) at (-0.500001, 0.275) {};
		\node [style=none] (191) at (-0.500001, -0.275) {};
		\node [style=none] (192) at (-0.125001, 0) {$\phi$};
		\node [style=none] (193) at (-1.075, 3) {};
		\node [style=none] (194) at (-1.075, 1.25) {};
		\node [style=none] (195) at (-1.075, -1.25) {};
		\node [style=none] (196) at (-1.075, -3) {};
		\node [style=none] (197) at (0.25, 0.75) {};
		\node [style=none] (198) at (0.25, -0.75) {};
		\node [style=none] (199) at (0.25, 0.275) {};
		\node [style=none] (200) at (0.25, -0.275) {};
		\node [style=none] (201) at (0.75, 1.975) {};
		\node [style=none] (202) at (0.25, 1.25) {};
		\node [style=none] (203) at (0.25, -1.25) {};
		\node [style=none] (204) at (0.75, -1.975) {};
		\node [style=none] (205) at (0.25, 1.25) {};
		\node [style=none] (206) at (0.5, -1.975) {};
		\node [style=none] (207) at (0.5, -2.525) {};
		\node [style=none] (208) at (1.75, -1.975) {};
		\node [style=none] (209) at (1.75, -2.525) {};
		\node [style=none] (210) at (0.75, -1.975) {};
		\node [style=none] (211) at (0.75, -2.525) {};
		\node [style=none] (212) at (1.125, -2.25) {$f$};
		\node [style=none] (213) at (1.5, -1.975) {};
		\node [style=none] (214) at (1.5, -2.525) {};
		\node [style=none] (215) at (0.5, 2.525) {};
		\node [style=none] (216) at (0.5, 1.975) {};
		\node [style=none] (217) at (1.75, 2.525) {};
		\node [style=none] (218) at (1.75, 1.975) {};
		\node [style=none] (219) at (0.75, 2.525) {};
		\node [style=none] (220) at (0.75, 1.975) {};
		\node [style=none] (221) at (1.125, 2.25) {$g$};
		\node [style=none] (222) at (1.5, 2.525) {};
		\node [style=none] (223) at (1.5, 1.975) {};
		\node [style=none] (224) at (1.5, 3) {};
		\node [style=none] (225) at (1.5, -3) {};
	\end{pgfonlayer}
	\begin{pgfonlayer}{edgelayer}
		\draw [style={supermap_shade}] (1.center)
			 to (10.center)
			 to (9.center)
			 to (2.center)
			 to cycle;
		\draw [style={hole_shade}] (54.center)
			 to (53.center)
			 to (52.center)
			 to (49.center)
			 to cycle;
		\draw [style={morphism_shade}] (184.center)
			 to (186.center)
			 to (187.center)
			 to (185.center)
			 to cycle;
		\draw (190.center) to (188.center);
		\draw (191.center) to (189.center);
		\draw (194.center) to (193.center);
		\draw (196.center) to (195.center);
		\draw (199.center) to (197.center);
		\draw (200.center) to (198.center);
		\draw [in=270, out=90] (202.center) to (201.center);
		\draw [in=270, out=90] (204.center) to (203.center);
		\draw [style=dottededge] (205.center) to (197.center);
		\draw [style=dottededge] (198.center) to (203.center);
		\draw [style={morphism_shade}] (206.center)
			 to (208.center)
			 to (209.center)
			 to (207.center)
			 to cycle;
		\draw [style={morphism_shade}] (215.center)
			 to (217.center)
			 to (218.center)
			 to (216.center)
			 to cycle;
		\draw (223.center) to (213.center);
		\draw (224.center) to (222.center);
		\draw (225.center) to (214.center);
	\end{pgfonlayer}
\end{tikzpicture} \ = \   \begin{tikzpicture}[scale=0.6, {every node/.style}={scale=0.8}]
	\begin{pgfonlayer}{nodelayer}
		\node [style=none] (1) at (-2, 2.5) {};
		\node [style=none] (2) at (2, 2.5) {};
		\node [style=none] (9) at (2, -2.5) {};
		\node [style=none] (10) at (-2, -2.5) {};
		\node [style=none] (11) at (-1.75, 0) {$S$};
		\node [style=none] (49) at (1.5, 2) {};
		\node [style=none] (52) at (-1.5, 2) {};
		\node [style=none] (53) at (-1.5, -2) {};
		\node [style=none] (54) at (1.5, -2) {};
		\node [style=none] (184) at (-1.25, 0.275) {};
		\node [style=none] (185) at (-1.25, -0.275) {};
		\node [style=none] (186) at (0, 0.275) {};
		\node [style=none] (187) at (0, -0.275) {};
		\node [style=none] (188) at (-1, 2) {};
		\node [style=none] (189) at (-1, -2) {};
		\node [style=none] (190) at (-1, 0.275) {};
		\node [style=none] (191) at (-1, -0.275) {};
		\node [style=none] (192) at (-0.625001, 0) {$\phi$};
		\node [style=none] (193) at (-1.575, 3) {};
		\node [style=none] (194) at (-1.575, 2.5) {};
		\node [style=none] (195) at (-1.575, -2.5) {};
		\node [style=none] (196) at (-1.575, -3) {};
		\node [style=none] (199) at (-0.25, 0.25) {};
		\node [style=none] (201) at (0.25, 1) {};
		\node [style=none] (202) at (-0.25, 0.275) {};
		\node [style=none] (203) at (-0.25, -0.275) {};
		\node [style=none] (204) at (0.25, -1) {};
		\node [style=none] (206) at (0, -0.975) {};
		\node [style=none] (207) at (0, -1.525) {};
		\node [style=none] (208) at (1.25, -0.975) {};
		\node [style=none] (209) at (1.25, -1.525) {};
		\node [style=none] (210) at (0.25, -0.975) {};
		\node [style=none] (211) at (0.25, -1.525) {};
		\node [style=none] (212) at (0.625, -1.25) {$f$};
		\node [style=none] (213) at (1, -0.975) {};
		\node [style=none] (214) at (1, -1.525) {};
		\node [style=none] (215) at (0, 1.525) {};
		\node [style=none] (216) at (0, 0.975) {};
		\node [style=none] (217) at (1.25, 1.525) {};
		\node [style=none] (218) at (1.25, 0.975) {};
		\node [style=none] (219) at (0.25, 1.525) {};
		\node [style=none] (220) at (0.25, 0.975) {};
		\node [style=none] (221) at (0.625, 1.25) {$g$};
		\node [style=none] (222) at (1, 1.525) {};
		\node [style=none] (223) at (1, 0.975) {};
		\node [style=none] (224) at (1, 2) {};
		\node [style=none] (225) at (1, -2) {};
		\node [style=none] (226) at (1, 2.5) {};
		\node [style=none] (227) at (1, 3) {};
		\node [style=none] (228) at (1, -2.5) {};
		\node [style=none] (229) at (1, -3) {};
	\end{pgfonlayer}
	\begin{pgfonlayer}{edgelayer}
		\draw [style={supermap_shade}] (1.center)
			 to (10.center)
			 to (9.center)
			 to (2.center)
			 to cycle;
		\draw [style={hole_shade}] (54.center)
			 to (53.center)
			 to (52.center)
			 to (49.center)
			 to cycle;
		\draw [style={morphism_shade}] (184.center)
			 to (186.center)
			 to (187.center)
			 to (185.center)
			 to cycle;
		\draw (190.center) to (188.center);
		\draw (191.center) to (189.center);
		\draw (194.center) to (193.center);
		\draw (196.center) to (195.center);
		\draw [in=270, out=90] (202.center) to (201.center);
		\draw [in=270, out=90] (204.center) to (203.center);
		\draw [style={morphism_shade}] (206.center)
			 to (208.center)
			 to (209.center)
			 to (207.center)
			 to cycle;
		\draw [style={morphism_shade}] (215.center)
			 to (217.center)
			 to (218.center)
			 to (216.center)
			 to cycle;
		\draw (223.center) to (213.center);
		\draw (224.center) to (222.center);
		\draw (225.center) to (214.center);
		\draw (229.center) to (228.center);
		\draw (227.center) to (226.center);
		\draw [style=dottededge] (226.center) to (224.center);
		\draw [style=dottededge] (225.center) to (228.center);
	\end{pgfonlayer}
\end{tikzpicture} . \]
Furthermore,  a multi-input locally-applicable tranformation of type \[    S:   \otimes_i (a_i \Rightarrow a_i') \rightarrow (b \Rightarrow b')  \] is a family of functions \[S_{x_1 x_1' x_2 x_2' }:  \times_i \cat{C}(a_i \otimes x_i , a_i'   \otimes x_i'  ) \rightarrow \cat{C}(b \otimes_i (x_i)  , b'   \otimes_i (x_i' ) )\] satisfying  
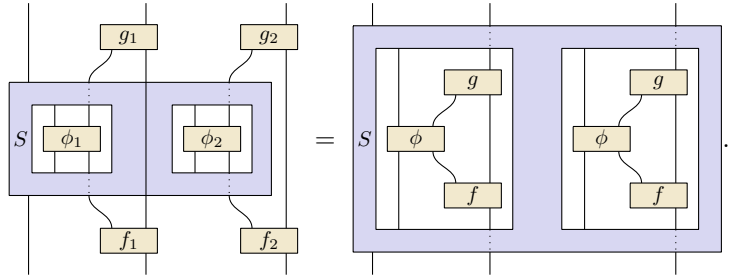
\[   \begin{tikzpicture}[scale=0.6, {every node/.style}={scale=0.8}]
	\begin{pgfonlayer}{nodelayer}
		\node [style=none] (1) at (-1.5, 1.25) {};
		\node [style=none] (2) at (4.25, 1.25) {};
		\node [style=none] (9) at (4.25, -1.25) {};
		\node [style=none] (10) at (-1.5, -1.25) {};
		\node [style=none] (11) at (-1.25, 0) {$S$};
		\node [style=none] (49) at (0.75, 0.75) {};
		\node [style=none] (52) at (-1, 0.75) {};
		\node [style=none] (53) at (-1, -0.75) {};
		\node [style=none] (54) at (0.75, -0.75) {};
		\node [style=none] (184) at (-0.750001, 0.275) {};
		\node [style=none] (185) at (-0.750001, -0.275) {};
		\node [style=none] (186) at (0.499999, 0.275) {};
		\node [style=none] (187) at (0.499999, -0.275) {};
		\node [style=none] (188) at (-0.500001, 0.75) {};
		\node [style=none] (189) at (-0.500001, -0.75) {};
		\node [style=none] (190) at (-0.500001, 0.275) {};
		\node [style=none] (191) at (-0.500001, -0.275) {};
		\node [style=none] (192) at (-0.125001, 0) {$\phi_1$};
		\node [style=none] (193) at (-1.075, 3) {};
		\node [style=none] (194) at (-1.075, 1.25) {};
		\node [style=none] (195) at (-1.075, -1.25) {};
		\node [style=none] (196) at (-1.075, -3) {};
		\node [style=none] (197) at (0.25, 0.75) {};
		\node [style=none] (198) at (0.25, -0.75) {};
		\node [style=none] (199) at (0.25, 0.275) {};
		\node [style=none] (200) at (0.25, -0.275) {};
		\node [style=none] (201) at (0.75, 1.975) {};
		\node [style=none] (202) at (0.25, 1.25) {};
		\node [style=none] (203) at (0.25, -1.25) {};
		\node [style=none] (204) at (0.75, -1.975) {};
		\node [style=none] (205) at (0.25, 1.25) {};
		\node [style=none] (206) at (0.5, -1.975) {};
		\node [style=none] (207) at (0.5, -2.525) {};
		\node [style=none] (208) at (1.75, -1.975) {};
		\node [style=none] (209) at (1.75, -2.525) {};
		\node [style=none] (210) at (0.75, -1.975) {};
		\node [style=none] (211) at (0.75, -2.525) {};
		\node [style=none] (212) at (1.125, -2.25) {$f_1$};
		\node [style=none] (213) at (1.5, -1.975) {};
		\node [style=none] (214) at (1.5, -2.525) {};
		\node [style=none] (215) at (0.5, 2.525) {};
		\node [style=none] (216) at (0.5, 1.975) {};
		\node [style=none] (217) at (1.75, 2.525) {};
		\node [style=none] (218) at (1.75, 1.975) {};
		\node [style=none] (219) at (0.75, 2.525) {};
		\node [style=none] (220) at (0.75, 1.975) {};
		\node [style=none] (221) at (1.125, 2.25) {$g_1$};
		\node [style=none] (222) at (1.5, 2.525) {};
		\node [style=none] (223) at (1.5, 1.975) {};
		\node [style=none] (224) at (1.5, 3) {};
		\node [style=none] (225) at (1.5, -3) {};
		\node [style=none] (227) at (4.325, -1.25) {};
		\node [style=none] (228) at (3.825, 0.75) {};
		\node [style=none] (229) at (2.075, 0.75) {};
		\node [style=none] (230) at (2.075, -0.75) {};
		\node [style=none] (231) at (3.825, -0.75) {};
		\node [style=none] (232) at (2.325, 0.275) {};
		\node [style=none] (233) at (2.325, -0.275) {};
		\node [style=none] (234) at (3.575, 0.275) {};
		\node [style=none] (235) at (3.575, -0.275) {};
		\node [style=none] (236) at (2.575, 0.75) {};
		\node [style=none] (237) at (2.575, -0.75) {};
		\node [style=none] (238) at (2.575, 0.275) {};
		\node [style=none] (239) at (2.575, -0.275) {};
		\node [style=none] (240) at (2.95, 0) {$\phi_2$};
		\node [style=none] (245) at (3.325, 0.75) {};
		\node [style=none] (246) at (3.325, -0.75) {};
		\node [style=none] (247) at (3.325, 0.275) {};
		\node [style=none] (248) at (3.325, -0.275) {};
		\node [style=none] (249) at (3.825, 1.975) {};
		\node [style=none] (250) at (3.325, 1.25) {};
		\node [style=none] (251) at (3.325, -1.25) {};
		\node [style=none] (252) at (3.825, -1.975) {};
		\node [style=none] (253) at (3.325, 1.25) {};
		\node [style=none] (254) at (3.575, -1.975) {};
		\node [style=none] (255) at (3.575, -2.525) {};
		\node [style=none] (256) at (4.825, -1.975) {};
		\node [style=none] (257) at (4.825, -2.525) {};
		\node [style=none] (258) at (3.825, -1.975) {};
		\node [style=none] (259) at (3.825, -2.525) {};
		\node [style=none] (260) at (4.2, -2.25) {$f_2$};
		\node [style=none] (261) at (4.575, -1.975) {};
		\node [style=none] (262) at (4.575, -2.525) {};
		\node [style=none] (263) at (3.575, 2.525) {};
		\node [style=none] (264) at (3.575, 1.975) {};
		\node [style=none] (265) at (4.825, 2.525) {};
		\node [style=none] (266) at (4.825, 1.975) {};
		\node [style=none] (267) at (3.825, 2.525) {};
		\node [style=none] (268) at (3.825, 1.975) {};
		\node [style=none] (269) at (4.2, 2.25) {$g_2$};
		\node [style=none] (270) at (4.575, 2.525) {};
		\node [style=none] (271) at (4.575, 1.975) {};
		\node [style=none] (272) at (4.575, 3) {};
		\node [style=none] (273) at (4.575, -3) {};
	\end{pgfonlayer}
	\begin{pgfonlayer}{edgelayer}
		\draw [style={supermap_shade}] (1.center)
			 to (10.center)
			 to (9.center)
			 to (2.center)
			 to cycle;
		\draw [style={hole_shade}] (54.center)
			 to (53.center)
			 to (52.center)
			 to (49.center)
			 to cycle;
		\draw [style={morphism_shade}] (184.center)
			 to (186.center)
			 to (187.center)
			 to (185.center)
			 to cycle;
		\draw (190.center) to (188.center);
		\draw (191.center) to (189.center);
		\draw (194.center) to (193.center);
		\draw (196.center) to (195.center);
		\draw (199.center) to (197.center);
		\draw (200.center) to (198.center);
		\draw [in=270, out=90] (202.center) to (201.center);
		\draw [in=270, out=90] (204.center) to (203.center);
		\draw [style=dottededge] (205.center) to (197.center);
		\draw [style=dottededge] (198.center) to (203.center);
		\draw [style={morphism_shade}] (206.center)
			 to (208.center)
			 to (209.center)
			 to (207.center)
			 to cycle;
		\draw [style={morphism_shade}] (215.center)
			 to (217.center)
			 to (218.center)
			 to (216.center)
			 to cycle;
		\draw (223.center) to (213.center);
		\draw (224.center) to (222.center);
		\draw (225.center) to (214.center);
		\draw [style={hole_shade}] (231.center)
			 to (230.center)
			 to (229.center)
			 to (228.center)
			 to cycle;
		\draw [style={morphism_shade}] (232.center)
			 to (234.center)
			 to (235.center)
			 to (233.center)
			 to cycle;
		\draw (238.center) to (236.center);
		\draw (239.center) to (237.center);
		\draw (247.center) to (245.center);
		\draw (248.center) to (246.center);
		\draw [in=270, out=90] (250.center) to (249.center);
		\draw [in=270, out=90] (252.center) to (251.center);
		\draw [style=dottededge] (253.center) to (245.center);
		\draw [style=dottededge] (246.center) to (251.center);
		\draw [style={morphism_shade}] (254.center)
			 to (256.center)
			 to (257.center)
			 to (255.center)
			 to cycle;
		\draw [style={morphism_shade}] (263.center)
			 to (265.center)
			 to (266.center)
			 to (264.center)
			 to cycle;
		\draw (271.center) to (261.center);
		\draw (272.center) to (270.center);
		\draw (273.center) to (262.center);
	\end{pgfonlayer}
\end{tikzpicture} \ = \ \begin{tikzpicture}[scale=0.6, {every node/.style}={scale=0.8}]
	\begin{pgfonlayer}{nodelayer}
		\node [style=none] (1) at (-4.25, 2.5) {};
		\node [style=none] (2) at (3.825, 2.5) {};
		\node [style=none] (9) at (3.825, -2.5) {};
		\node [style=none] (10) at (-4.25, -2.5) {};
		\node [style=none] (11) at (-4, 0) {$S$};
		\node [style=none] (49) at (-0.75, 2) {};
		\node [style=none] (52) at (-3.75, 2) {};
		\node [style=none] (53) at (-3.75, -2) {};
		\node [style=none] (54) at (-0.75, -2) {};
		\node [style=none] (184) at (-3.5, 0.275) {};
		\node [style=none] (185) at (-3.5, -0.275) {};
		\node [style=none] (186) at (-2.25, 0.275) {};
		\node [style=none] (187) at (-2.25, -0.275) {};
		\node [style=none] (188) at (-3.25, 2) {};
		\node [style=none] (189) at (-3.25, -2) {};
		\node [style=none] (190) at (-3.25, 0.275) {};
		\node [style=none] (191) at (-3.25, -0.275) {};
		\node [style=none] (192) at (-2.875, 0) {$\phi$};
		\node [style=none] (193) at (-3.825, 3) {};
		\node [style=none] (194) at (-3.825, 2.5) {};
		\node [style=none] (195) at (-3.825, -2.5) {};
		\node [style=none] (196) at (-3.825, -3) {};
		\node [style=none] (201) at (-2, 0.975) {};
		\node [style=none] (202) at (-2.5, 0.275) {};
		\node [style=none] (203) at (-2.5, -0.275) {};
		\node [style=none] (204) at (-2, -0.975) {};
		\node [style=none] (206) at (-2.25, -0.975) {};
		\node [style=none] (207) at (-2.25, -1.525) {};
		\node [style=none] (208) at (-1, -0.975) {};
		\node [style=none] (209) at (-1, -1.525) {};
		\node [style=none] (211) at (-2, -1.525) {};
		\node [style=none] (212) at (-1.625, -1.25) {$f$};
		\node [style=none] (213) at (-1.25, -0.975) {};
		\node [style=none] (214) at (-1.25, -1.525) {};
		\node [style=none] (215) at (-2.25, 1.525) {};
		\node [style=none] (216) at (-2.25, 0.975) {};
		\node [style=none] (217) at (-1, 1.525) {};
		\node [style=none] (218) at (-1, 0.975) {};
		\node [style=none] (219) at (-2, 1.525) {};
		\node [style=none] (221) at (-1.625, 1.25) {$g$};
		\node [style=none] (222) at (-1.25, 1.525) {};
		\node [style=none] (223) at (-1.25, 0.975) {};
		\node [style=none] (224) at (-1.25, 2) {};
		\node [style=none] (225) at (-1.25, -2) {};
		\node [style=none] (226) at (-1.25, 2.5) {};
		\node [style=none] (227) at (-1.25, 3) {};
		\node [style=none] (228) at (-1.25, -2.5) {};
		\node [style=none] (229) at (-1.25, -3) {};
		\node [style=none] (230) at (3.325, 2) {};
		\node [style=none] (231) at (0.325, 2) {};
		\node [style=none] (232) at (0.325, -2) {};
		\node [style=none] (233) at (3.325, -2) {};
		\node [style=none] (234) at (0.575, 0.275) {};
		\node [style=none] (235) at (0.575, -0.275) {};
		\node [style=none] (236) at (1.825, 0.275) {};
		\node [style=none] (237) at (1.825, -0.275) {};
		\node [style=none] (238) at (0.825, 2) {};
		\node [style=none] (239) at (0.825, -2) {};
		\node [style=none] (240) at (0.825, 0.275) {};
		\node [style=none] (241) at (0.825, -0.275) {};
		\node [style=none] (242) at (1.2, 0) {$\phi$};
		\node [style=none] (249) at (2.075, 0.975) {};
		\node [style=none] (250) at (1.575, 0.275) {};
		\node [style=none] (251) at (1.575, -0.275) {};
		\node [style=none] (252) at (2.075, -0.975) {};
		\node [style=none] (254) at (1.825, -0.975) {};
		\node [style=none] (255) at (1.825, -1.525) {};
		\node [style=none] (256) at (3.075, -0.975) {};
		\node [style=none] (257) at (3.075, -1.525) {};
		\node [style=none] (259) at (2.075, -1.525) {};
		\node [style=none] (260) at (2.45, -1.25) {$f$};
		\node [style=none] (261) at (2.825, -0.975) {};
		\node [style=none] (262) at (2.825, -1.525) {};
		\node [style=none] (263) at (1.825, 1.525) {};
		\node [style=none] (264) at (1.825, 0.975) {};
		\node [style=none] (265) at (3.075, 1.525) {};
		\node [style=none] (266) at (3.075, 0.975) {};
		\node [style=none] (267) at (2.075, 1.525) {};
		\node [style=none] (269) at (2.45, 1.25) {$g$};
		\node [style=none] (270) at (2.825, 1.525) {};
		\node [style=none] (271) at (2.825, 0.975) {};
		\node [style=none] (272) at (2.825, 2) {};
		\node [style=none] (273) at (2.825, -2) {};
		\node [style=none] (274) at (2.825, 2.5) {};
		\node [style=none] (275) at (2.825, 3) {};
		\node [style=none] (276) at (2.825, -2.5) {};
		\node [style=none] (277) at (2.825, -3) {};
	\end{pgfonlayer}
	\begin{pgfonlayer}{edgelayer}
		\draw [style={supermap_shade}] (1.center)
			 to (10.center)
			 to (9.center)
			 to (2.center)
			 to cycle;
		\draw [style={hole_shade}] (54.center)
			 to (53.center)
			 to (52.center)
			 to (49.center)
			 to cycle;
		\draw [style={morphism_shade}] (184.center)
			 to (186.center)
			 to (187.center)
			 to (185.center)
			 to cycle;
		\draw (190.center) to (188.center);
		\draw (191.center) to (189.center);
		\draw (194.center) to (193.center);
		\draw (196.center) to (195.center);
		\draw [in=270, out=90] (202.center) to (201.center);
		\draw [in=270, out=90] (204.center) to (203.center);
		\draw [style={morphism_shade}] (206.center)
			 to (208.center)
			 to (209.center)
			 to (207.center)
			 to cycle;
		\draw [style={morphism_shade}] (215.center)
			 to (217.center)
			 to (218.center)
			 to (216.center)
			 to cycle;
		\draw (223.center) to (213.center);
		\draw (224.center) to (222.center);
		\draw (225.center) to (214.center);
		\draw (229.center) to (228.center);
		\draw (227.center) to (226.center);
		\draw [style=dottededge] (226.center) to (224.center);
		\draw [style=dottededge] (225.center) to (228.center);
		\draw [style={hole_shade}] (233.center)
			 to (232.center)
			 to (231.center)
			 to (230.center)
			 to cycle;
		\draw [style={morphism_shade}] (234.center)
			 to (236.center)
			 to (237.center)
			 to (235.center)
			 to cycle;
		\draw (240.center) to (238.center);
		\draw (241.center) to (239.center);
		\draw [in=270, out=90] (250.center) to (249.center);
		\draw [in=270, out=90] (252.center) to (251.center);
		\draw [style={morphism_shade}] (254.center)
			 to (256.center)
			 to (257.center)
			 to (255.center)
			 to cycle;
		\draw [style={morphism_shade}] (263.center)
			 to (265.center)
			 to (266.center)
			 to (264.center)
			 to cycle;
		\draw (271.center) to (261.center);
		\draw (272.center) to (270.center);
		\draw (273.center) to (262.center);
		\draw (277.center) to (276.center);
		\draw (275.center) to (274.center);
		\draw [style=dottededge] (274.center) to (272.center);
		\draw [style=dottededge] (273.center) to (276.center);
	\end{pgfonlayer}
\end{tikzpicture}.    \]
\end{definition}

It has been proven that the quantum supermaps \cite{Chiribella2008TransformingSupermaps}  (or as a special case, the process matrices \cite{Oreshkov2012QuantumOrder}) are in one-to-one correspondence with the locally-applciable transformations on the symmetric monoidal category of quantum channels \cite{wilson2025quantumsupermapscharacterizedlocality}. For any symmetric monoidal category $\mathcal{C}$ it is straightforward to see that $\mathcal{C}$ is symmetric monoidal enriched in a multi-category of locally-applicable transformations $\mathbf{Lot}[\mathcal{C}]$ as outlined in \cite{wilson2025quantumsupermapscharacterizedlocality}.
In-fact, the theory of locally-applicable transformations can be identified as a fragment of the theory of strong profunctors \cite{hefford_supermaps}, which gives a general construction technique for constructing duoidal \cite{hefford_supermaps} and BV \cite{hefford2025bvcategoryspacetimeinterventions} categories of higher-order processes. 

The locally-applicable transformations are actually quite poorly behaved compositionally, and do not in general form even a symmetric polycategory. There are, however, procedures for carving out symmetric polycategories from the locally-applicable transformations.
\begin{example}
In general the locally-applicable transformations fail to form a polycategory because of the failure of commutation of locally-applicable transformations when applied to parts of bipartite processes \cite{wilson_polycategories}. This can be remedied by keeping only those locally-applicable transformations which when reduced to a single hole become central, that is, only those transformations which do commute with any other locally-applicable transformation are kept. These transformations were referred to in \cite{wilson_polycategories} as slots, and were proven to recover the quantum superunitaries when applied to the monoidal category of unitary linear maps. 

More precisely, a slot (strongly locally-applicable transformation) \[S: (a \Rightarrow a') \rightarrow (b \Rightarrow b'), \] is a locally-applicable transformation of the same type such that for any locally applicable transformation \[T: (c \Rightarrow c') \rightarrow (d \Rightarrow d'), \] then 
 
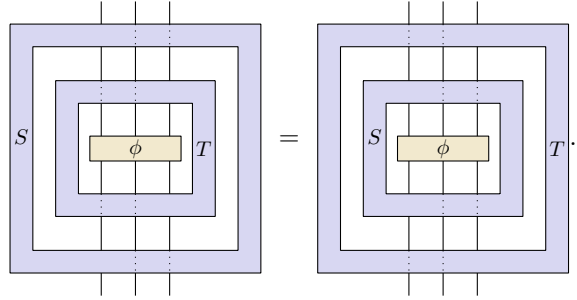
\[   \begin{tikzpicture}[scale=0.6, {every node/.style}={scale=0.8}]
	\begin{pgfonlayer}{nodelayer}
		\node [style=none] (1) at (-3, 2.5) {};
		\node [style=none] (2) at (2.5, 2.5) {};
		\node [style=none] (9) at (2.5, -3) {};
		\node [style=none] (10) at (-3, -3) {};
		\node [style=none] (11) at (-2.75, 0) {$S$};
		\node [style=none] (49) at (2, 2) {};
		\node [style=none] (52) at (-2.5, 2) {};
		\node [style=none] (53) at (-2.5, -2.5) {};
		\node [style=none] (54) at (2, -2.5) {};
		\node [style=none] (188) at (-1, 2) {};
		\node [style=none] (189) at (-1, -2.5) {};
		\node [style=none] (193) at (-1, 3) {};
		\node [style=none] (194) at (-1, 2.5) {};
		\node [style=none] (195) at (-1, -3) {};
		\node [style=none] (196) at (-1, -3.5) {};
		\node [style=none] (201) at (-0.25, 2) {};
		\node [style=none] (204) at (-0.25, -2.5) {};
		\node [style=none] (224) at (-0.25, 2) {};
		\node [style=none] (225) at (-0.25, -2.5) {};
		\node [style=none] (226) at (-0.25, 2.5) {};
		\node [style=none] (227) at (-0.25, 3) {};
		\node [style=none] (228) at (-0.25, -3) {};
		\node [style=none] (229) at (-0.25, -3.5) {};
		\node [style=none] (231) at (0.5, 2) {};
		\node [style=none] (234) at (0.5, -2.5) {};
		\node [style=none] (235) at (0.5, 2) {};
		\node [style=none] (236) at (0.5, -2.5) {};
		\node [style=none] (237) at (0.5, 2.5) {};
		\node [style=none] (238) at (0.5, 3) {};
		\node [style=none] (239) at (0.5, -3) {};
		\node [style=none] (240) at (0.5, -3.5) {};
		\node [style=none] (241) at (1.5, 1.25) {};
		\node [style=none] (242) at (-2, 1.25) {};
		\node [style=none] (243) at (-2, -1.75) {};
		\node [style=none] (244) at (1.5, -1.75) {};
		\node [style=none] (245) at (1.25, -0.25) {$T$};
		\node [style=none] (246) at (-1.5, 0.75) {};
		\node [style=none] (247) at (1, 0.75) {};
		\node [style=none] (248) at (1, -1.25) {};
		\node [style=none] (249) at (-1.5, -1.25) {};
		\node [style=none] (250) at (0.75, 0.025) {};
		\node [style=none] (251) at (0.75, -0.525) {};
		\node [style=none] (252) at (-1.25, 0.025) {};
		\node [style=none] (253) at (-1.25, -0.525) {};
		\node [style=none] (254) at (0.5, 0.75) {};
		\node [style=none] (255) at (0.5, -1.25) {};
		\node [style=none] (256) at (0.5, 0.025) {};
		\node [style=none] (257) at (0.5, -0.525) {};
		\node [style=none] (258) at (-0.249999, -0.25) {$\phi$};
		\node [style=none] (259) at (0.5, 2) {};
		\node [style=none] (260) at (0.5, 1.25) {};
		\node [style=none] (261) at (0.5, -1.75) {};
		\node [style=none] (262) at (0.5, -2.5) {};
		\node [style=none] (263) at (-0.25, 0) {};
		\node [style=none] (264) at (-0.25, 0.75) {};
		\node [style=none] (265) at (-0.25, 0.025) {};
		\node [style=none] (266) at (-0.25, -0.525) {};
		\node [style=none] (267) at (-0.25, -1.25) {};
		\node [style=none] (268) at (-0.25, 0.75) {};
		\node [style=none] (269) at (-0.25, -1.25) {};
		\node [style=none] (270) at (-0.25, 1.25) {};
		\node [style=none] (271) at (-0.25, 2) {};
		\node [style=none] (272) at (-0.25, -1.75) {};
		\node [style=none] (273) at (-0.25, -2.5) {};
		\node [style=none] (274) at (-1, 0) {};
		\node [style=none] (275) at (-1, 0.75) {};
		\node [style=none] (276) at (-1, 0.025) {};
		\node [style=none] (277) at (-1, -0.525) {};
		\node [style=none] (278) at (-1, -1.25) {};
		\node [style=none] (279) at (-1, 0.75) {};
		\node [style=none] (280) at (-1, -1.25) {};
		\node [style=none] (281) at (-1, 1.25) {};
		\node [style=none] (282) at (-1, 2) {};
		\node [style=none] (283) at (-1, -1.75) {};
		\node [style=none] (284) at (-1, -2.5) {};
	\end{pgfonlayer}
	\begin{pgfonlayer}{edgelayer}
		\draw [style={supermap_shade}] (1.center)
			 to (10.center)
			 to (9.center)
			 to (2.center)
			 to cycle;
		\draw [style={hole_shade}] (54.center)
			 to (53.center)
			 to (52.center)
			 to (49.center)
			 to cycle;
		\draw (194.center) to (193.center);
		\draw (196.center) to (195.center);
		\draw (229.center) to (228.center);
		\draw (227.center) to (226.center);
		\draw [style=dottededge] (226.center) to (224.center);
		\draw [style=dottededge] (225.center) to (228.center);
		\draw (240.center) to (239.center);
		\draw (238.center) to (237.center);
		\draw [style=dottededge] (237.center) to (235.center);
		\draw [style=dottededge] (236.center) to (239.center);
		\draw [style={supermap_shade}] (241.center)
			 to (244.center)
			 to (243.center)
			 to (242.center)
			 to cycle;
		\draw [style={hole_shade}] (249.center)
			 to (248.center)
			 to (247.center)
			 to (246.center)
			 to cycle;
		\draw [style={morphism_shade}] (250.center)
			 to (252.center)
			 to (253.center)
			 to (251.center)
			 to cycle;
		\draw (256.center) to (254.center);
		\draw (257.center) to (255.center);
		\draw (260.center) to (259.center);
		\draw (262.center) to (261.center);
		\draw [in=-90, out=90] (265.center) to (264.center);
		\draw [in=-90, out=90] (267.center) to (266.center);
		\draw (273.center) to (272.center);
		\draw (271.center) to (270.center);
		\draw [style=dottededge] (270.center) to (268.center);
		\draw [style=dottededge] (269.center) to (272.center);
		\draw [in=-90, out=90] (276.center) to (275.center);
		\draw [in=-90, out=90] (278.center) to (277.center);
		\draw (284.center) to (283.center);
		\draw (282.center) to (281.center);
		\draw [style=dottededge] (281.center) to (279.center);
		\draw [style=dottededge] (280.center) to (283.center);
	\end{pgfonlayer}
\end{tikzpicture} \ = \ \begin{tikzpicture}[scale=0.6, {every node/.style}={scale=0.8}]
	\begin{pgfonlayer}{nodelayer}
		\node [style=none] (1) at (2.5, 2.5) {};
		\node [style=none] (2) at (-3, 2.5) {};
		\node [style=none] (9) at (-3, -3) {};
		\node [style=none] (10) at (2.5, -3) {};
		\node [style=none] (11) at (-1.75, 0) {$S$};
		\node [style=none] (49) at (-2.5, 2) {};
		\node [style=none] (52) at (2, 2) {};
		\node [style=none] (53) at (2, -2.5) {};
		\node [style=none] (54) at (-2.5, -2.5) {};
		\node [style=none] (188) at (0.5, 2) {};
		\node [style=none] (189) at (0.5, -2.5) {};
		\node [style=none] (193) at (0.5, 3) {};
		\node [style=none] (194) at (0.5, 2.5) {};
		\node [style=none] (195) at (0.5, -3) {};
		\node [style=none] (196) at (0.5, -3.5) {};
		\node [style=none] (201) at (-0.25, 2) {};
		\node [style=none] (204) at (-0.25, -2.5) {};
		\node [style=none] (224) at (-0.25, 2) {};
		\node [style=none] (225) at (-0.25, -2.5) {};
		\node [style=none] (226) at (-0.25, 2.5) {};
		\node [style=none] (227) at (-0.25, 3) {};
		\node [style=none] (228) at (-0.25, -3) {};
		\node [style=none] (229) at (-0.25, -3.5) {};
		\node [style=none] (231) at (-1, 2) {};
		\node [style=none] (234) at (-1, -2.5) {};
		\node [style=none] (235) at (-1, 2) {};
		\node [style=none] (236) at (-1, -2.5) {};
		\node [style=none] (237) at (-1, 2.5) {};
		\node [style=none] (238) at (-1, 3) {};
		\node [style=none] (239) at (-1, -3) {};
		\node [style=none] (240) at (-1, -3.5) {};
		\node [style=none] (241) at (-2, 1.25) {};
		\node [style=none] (242) at (1.5, 1.25) {};
		\node [style=none] (243) at (1.5, -1.75) {};
		\node [style=none] (244) at (-2, -1.75) {};
		\node [style=none] (245) at (2.25, -0.25) {$T$};
		\node [style=none] (246) at (1, 0.75) {};
		\node [style=none] (247) at (-1.5, 0.75) {};
		\node [style=none] (248) at (-1.5, -1.25) {};
		\node [style=none] (249) at (1, -1.25) {};
		\node [style=none] (250) at (-1.25, 0.025) {};
		\node [style=none] (251) at (-1.25, -0.525) {};
		\node [style=none] (252) at (0.75, 0.025) {};
		\node [style=none] (253) at (0.75, -0.525) {};
		\node [style=none] (254) at (-1, 0.75) {};
		\node [style=none] (255) at (-1, -1.25) {};
		\node [style=none] (256) at (-1, 0.025) {};
		\node [style=none] (257) at (-1, -0.525) {};
		\node [style=none] (258) at (-0.250001, -0.25) {$\phi$};
		\node [style=none] (259) at (-1, 2) {};
		\node [style=none] (260) at (-1, 1.25) {};
		\node [style=none] (261) at (-1, -1.75) {};
		\node [style=none] (262) at (-1, -2.5) {};
		\node [style=none] (263) at (-0.25, 0) {};
		\node [style=none] (264) at (-0.25, 0.75) {};
		\node [style=none] (265) at (-0.25, 0.025) {};
		\node [style=none] (266) at (-0.25, -0.525) {};
		\node [style=none] (267) at (-0.25, -1.25) {};
		\node [style=none] (268) at (-0.25, 0.75) {};
		\node [style=none] (269) at (-0.25, -1.25) {};
		\node [style=none] (270) at (-0.25, 1.25) {};
		\node [style=none] (271) at (-0.25, 2) {};
		\node [style=none] (272) at (-0.25, -1.75) {};
		\node [style=none] (273) at (-0.25, -2.5) {};
		\node [style=none] (274) at (0.5, 0) {};
		\node [style=none] (275) at (0.5, 0.75) {};
		\node [style=none] (276) at (0.5, 0.025) {};
		\node [style=none] (277) at (0.5, -0.525) {};
		\node [style=none] (278) at (0.5, -1.25) {};
		\node [style=none] (279) at (0.5, 0.75) {};
		\node [style=none] (280) at (0.5, -1.25) {};
		\node [style=none] (281) at (0.5, 1.25) {};
		\node [style=none] (282) at (0.5, 2) {};
		\node [style=none] (283) at (0.5, -1.75) {};
		\node [style=none] (284) at (0.5, -2.5) {};
	\end{pgfonlayer}
	\begin{pgfonlayer}{edgelayer}
		\draw [style={supermap_shade}] (1.center)
			 to (10.center)
			 to (9.center)
			 to (2.center)
			 to cycle;
		\draw [style={hole_shade}] (54.center)
			 to (53.center)
			 to (52.center)
			 to (49.center)
			 to cycle;
		\draw (194.center) to (193.center);
		\draw (196.center) to (195.center);
		\draw (229.center) to (228.center);
		\draw (227.center) to (226.center);
		\draw [style=dottededge] (226.center) to (224.center);
		\draw [style=dottededge] (225.center) to (228.center);
		\draw (240.center) to (239.center);
		\draw (238.center) to (237.center);
		\draw [style=dottededge] (237.center) to (235.center);
		\draw [style=dottededge] (236.center) to (239.center);
		\draw [style={supermap_shade}] (241.center)
			 to (244.center)
			 to (243.center)
			 to (242.center)
			 to cycle;
		\draw [style={hole_shade}] (249.center)
			 to (248.center)
			 to (247.center)
			 to (246.center)
			 to cycle;
		\draw [style={morphism_shade}] (250.center)
			 to (252.center)
			 to (253.center)
			 to (251.center)
			 to cycle;
		\draw (256.center) to (254.center);
		\draw (257.center) to (255.center);
		\draw (260.center) to (259.center);
		\draw (262.center) to (261.center);
		\draw [in=-90, out=90] (265.center) to (264.center);
		\draw [in=-90, out=90] (267.center) to (266.center);
		\draw (273.center) to (272.center);
		\draw (271.center) to (270.center);
		\draw [style=dottededge] (270.center) to (268.center);
		\draw [style=dottededge] (269.center) to (272.center);
		\draw [in=-90, out=90] (276.center) to (275.center);
		\draw [in=-90, out=90] (278.center) to (277.center);
		\draw (284.center) to (283.center);
		\draw (282.center) to (281.center);
		\draw [style=dottededge] (281.center) to (279.center);
		\draw [style=dottededge] (280.center) to (283.center);
	\end{pgfonlayer}
\end{tikzpicture}.    \]
Multi-input slots (with $N$-holes) are then defined inductively, as those such that when any one hole is filled what remains is an $N-1$-hole slot. Such slots were proven in \cite{wilson_polycategories} to form a polycategory, and it is easy to see that all of the maps used to define enrichment of $\mathcal{C}$ into locally-applicable transformations are themselves slots, and so, that $\mathcal{C}$ is enriched in the associated polycategory $\mathbf{pslot}[\mathcal{C}]$. 

What remains is to construct the cotensor  
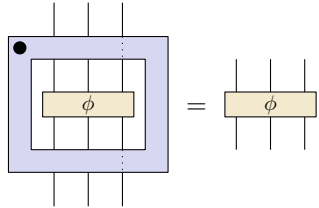
\[   \begin{tikzpicture}[scale=0.6, {every node/.style}={scale=0.8}]
	\begin{pgfonlayer}{nodelayer}
		\node [style={black_dot}] (11) at (-1.75, 1.25) {};
		\node [style=none] (188) at (0.5, 2.25) {};
		\node [style=none] (189) at (0.5, -2.25) {};
		\node [style=none] (201) at (-0.25, 2.25) {};
		\node [style=none] (204) at (-0.25, -2.25) {};
		\node [style=none] (224) at (-0.25, 2.25) {};
		\node [style=none] (225) at (-0.25, -2.25) {};
		\node [style=none] (231) at (-1, 2.25) {};
		\node [style=none] (234) at (-1, -2.25) {};
		\node [style=none] (235) at (-1, 2.25) {};
		\node [style=none] (236) at (-1, -2.25) {};
		\node [style=none] (241) at (-2, 1.5) {};
		\node [style=none] (242) at (1.5, 1.5) {};
		\node [style=none] (243) at (1.5, -1.5) {};
		\node [style=none] (244) at (-2, -1.5) {};
		\node [style=none] (246) at (1, 1) {};
		\node [style=none] (247) at (-1.5, 1) {};
		\node [style=none] (248) at (-1.5, -1) {};
		\node [style=none] (249) at (1, -1) {};
		\node [style=none] (250) at (-1.25, 0.275) {};
		\node [style=none] (251) at (-1.25, -0.275) {};
		\node [style=none] (252) at (0.75, 0.275) {};
		\node [style=none] (253) at (0.75, -0.275) {};
		\node [style=none] (254) at (-1, 1) {};
		\node [style=none] (255) at (-1, -1) {};
		\node [style=none] (256) at (-1, 0.275) {};
		\node [style=none] (257) at (-1, -0.275) {};
		\node [style=none] (258) at (-0.250001, 0) {$\phi$};
		\node [style=none] (259) at (-1, 2.25) {};
		\node [style=none] (260) at (-1, 1.5) {};
		\node [style=none] (261) at (-1, -1.5) {};
		\node [style=none] (262) at (-1, -2.25) {};
		\node [style=none] (263) at (-0.25, 0.25) {};
		\node [style=none] (264) at (-0.25, 1) {};
		\node [style=none] (265) at (-0.25, 0.275) {};
		\node [style=none] (266) at (-0.25, -0.275) {};
		\node [style=none] (267) at (-0.25, -1) {};
		\node [style=none] (268) at (-0.25, 1) {};
		\node [style=none] (269) at (-0.25, -1) {};
		\node [style=none] (270) at (-0.25, 1.5) {};
		\node [style=none] (271) at (-0.25, 2.25) {};
		\node [style=none] (272) at (-0.25, -1.5) {};
		\node [style=none] (273) at (-0.25, -2.25) {};
		\node [style=none] (274) at (0.5, 0.25) {};
		\node [style=none] (275) at (0.5, 1) {};
		\node [style=none] (276) at (0.5, 0.275) {};
		\node [style=none] (277) at (0.5, -0.275) {};
		\node [style=none] (278) at (0.5, -1) {};
		\node [style=none] (279) at (0.5, 1) {};
		\node [style=none] (280) at (0.5, -1) {};
		\node [style=none] (281) at (0.5, 1.5) {};
		\node [style=none] (282) at (0.5, 2.25) {};
		\node [style=none] (283) at (0.5, -1.5) {};
		\node [style=none] (284) at (0.5, -2.25) {};
	\end{pgfonlayer}
	\begin{pgfonlayer}{edgelayer}
		\draw [style={supermap_shade}] (241.center)
			 to (244.center)
			 to (243.center)
			 to (242.center)
			 to cycle;
		\draw [style={hole_shade}] (249.center)
			 to (248.center)
			 to (247.center)
			 to (246.center)
			 to cycle;
		\draw [style={morphism_shade}] (250.center)
			 to (252.center)
			 to (253.center)
			 to (251.center)
			 to cycle;
		\draw (256.center) to (254.center);
		\draw (257.center) to (255.center);
		\draw (260.center) to (259.center);
		\draw (262.center) to (261.center);
		\draw [in=-90, out=90] (265.center) to (264.center);
		\draw [in=-90, out=90] (267.center) to (266.center);
		\draw (273.center) to (272.center);
		\draw (271.center) to (270.center);
		\draw [in=-90, out=90] (276.center) to (275.center);
		\draw [in=-90, out=90] (278.center) to (277.center);
		\draw (284.center) to (283.center);
		\draw (282.center) to (281.center);
		\draw [style=dottededge] (281.center) to (279.center);
		\draw [style=dottededge] (280.center) to (283.center);
	\end{pgfonlayer}
\end{tikzpicture} \ = \ \begin{tikzpicture}[scale=0.6, {every node/.style}={scale=0.8}]
	\begin{pgfonlayer}{nodelayer}
		\node [style=none] (250) at (-1.25, 0.275) {};
		\node [style=none] (251) at (-1.25, -0.275) {};
		\node [style=none] (252) at (0.75, 0.275) {};
		\node [style=none] (253) at (0.75, -0.275) {};
		\node [style=none] (254) at (-1, 1) {};
		\node [style=none] (255) at (-1, -1) {};
		\node [style=none] (256) at (-1, 0.275) {};
		\node [style=none] (257) at (-1, -0.275) {};
		\node [style=none] (258) at (-0.250001, 0) {$\phi$};
		\node [style=none] (263) at (-0.25, 0.25) {};
		\node [style=none] (264) at (-0.25, 1) {};
		\node [style=none] (265) at (-0.25, 0.275) {};
		\node [style=none] (266) at (-0.25, -0.275) {};
		\node [style=none] (267) at (-0.25, -1) {};
		\node [style=none] (268) at (-0.25, 1) {};
		\node [style=none] (269) at (-0.25, -1) {};
		\node [style=none] (274) at (0.5, 0.25) {};
		\node [style=none] (275) at (0.5, 1) {};
		\node [style=none] (276) at (0.5, 0.275) {};
		\node [style=none] (277) at (0.5, -0.275) {};
		\node [style=none] (278) at (0.5, -1) {};
		\node [style=none] (279) at (0.5, 1) {};
		\node [style=none] (280) at (0.5, -1) {};
	\end{pgfonlayer}
	\begin{pgfonlayer}{edgelayer}
		\draw [style={morphism_shade}] (250.center)
			 to (252.center)
			 to (253.center)
			 to (251.center)
			 to cycle;
		\draw (256.center) to (254.center);
		\draw (257.center) to (255.center);
		\draw [in=-90, out=90] (265.center) to (264.center);
		\draw [in=-90, out=90] (267.center) to (266.center);
		\draw [in=-90, out=90] (276.center) to (275.center);
		\draw [in=-90, out=90] (278.center) to (277.center);
	\end{pgfonlayer}
\end{tikzpicture},    \] as the natural transformation in which every component is the identity function. The fact that each component is the identity function makes it immediate that the cotensor satisfies associativity, the frobenius law, and the copy law, and so the slots indeed form a higher-order circuit theory. 
\end{example}

\begin{example}
The polycategory of single-party respresentable supermaps \cite{wilson_polycategories} is the sub-polycategory of the polycategory of polyslots in which only the locally-applicable transformations which decompose locally as combs are kept. More precisely, a single-party representable supermap of type \[ S: [a_1,a_1'] \dots [a_n,a_n'] \rightarrow [b,b']  \] is a family of functions \[S_{x_1 \dots x_n,x_1 ' \dots x_n'}: \mathbf{C}(a_1 x_1,a_1 ' x_1 ') \dots \mathbf{C}(a_n x_n,a_n ' x_n ') \rightarrow \mathbf{C}(b x_1 \dots x_n, b_1 ' x_1'  \dots x_n')\] such that for every $i$ and family of morphisms $\phi_{(m)}$ with $m \in \{  1 \dots (i-1)(i+1) \dots n \}$ there exists $S(\phi_{(m)})_i^{u}$ and $S(\phi_{(m)})_i^{d}$ satisfying \[S_{x_1 \dots x_n,x_1 \dots x_n'}(\phi_1 \dots \phi_i \dots \phi_n) \quad  = \quad  \begin{tikzpicture}[scale=1, {every node/.style}={scale=1}]
	\begin{pgfonlayer}{nodelayer}
		\node [style=none] (1048) at (-0.5, 1) {};
		\node [style=none] (1059) at (0.5, 2) {};
		\node [style=none] (1061) at (-0.5, -0.25) {};
		\node [style=none] (1071) at (0.25, -0.25) {};
		\node [style=none] (1072) at (-0.5, -0.25) {};
		\node [style=none] (1074) at (-0.5, 1) {};
		\node [style=none] (1076) at (-0.5, -0.25) {};
		\node [style=none] (1077) at (-0.5, -1) {};
		\node [style=none] (1078) at (0.5, -2) {};
		\node [style=none] (1079) at (-0.5, -0.25) {};
		\node [style=none] (1081) at (-0.5, 0.25) {};
		\node [style=none] (1083) at (-0.5, -0.25) {};
		\node [style=none] (1085) at (-0.5, -0.25) {};
		\node [style=none] (1086) at (-0.5, -0.25) {};
		\node [style=none] (1089) at (-0.5, 0) {};
		\node [style=none] (1091) at (0.5, 0.25) {};
		\node [style=none] (1093) at (-0.75, 0.25) {};
		\node [style=none] (1094) at (0.75, 0.25) {};
		\node [style=none] (1095) at (0.75, -0.25) {};
		\node [style=none] (1096) at (-0.75, -0.25) {};
		\node [style=none] (1097) at (0, 0) {$\rho_k$};
		\node [style=none] (1098) at (0.5, -0.25) {};
		\node [style=none] (1099) at (0.5, 0.25) {};
		\node [style=none] (1100) at (0.5, -0.25) {};
		\node [style=none] (1101) at (0.5, 0.25) {};
		\node [style=none] (1102) at (0.5, 0.25) {};
		\node [style=none] (1104) at (-0.5, -0.25) {};
		\node [style=none] (1105) at (0.25, 1) {};
		\node [style=none] (1106) at (0.25, 1.5) {};
		\node [style=none] (1107) at (-1.75, 1.5) {};
		\node [style=none] (1108) at (-1.75, 1) {};
		\node [style=none] (1110) at (-1.25, 2) {};
		\node [style=none] (1111) at (-1.25, 1.5) {};
		\node [style=none] (1112) at (-1.25, 1) {};
		\node [style=none] (1113) at (-1.25, -1) {};
		\node [style=none] (1114) at (-0.5, -1) {};
		\node [style=none] (1115) at (-0.5, -1) {};
		\node [style=none] (1116) at (0.25, -1) {};
		\node [style=none] (1117) at (0.25, -1.5) {};
		\node [style=none] (1118) at (-1.75, -1.5) {};
		\node [style=none] (1119) at (-1.75, -1) {};
		\node [style=none] (1121) at (-1.25, -2) {};
		\node [style=none] (1122) at (-1.25, -1.5) {};
		\node [style=none] (1123) at (-1.25, -1) {};
		\node [style=none] (1124) at (1.25, 2) {};
		\node [style=none] (1125) at (-0.25, 1.5) {};
		\node [style=none] (1126) at (-0.5, 2) {};
		\node [style=none] (1127) at (-0.5, 1.5) {};
		\node [style=none] (1128) at (-0.5, -2) {};
		\node [style=none] (1129) at (1.25, -2) {};
		\node [style=none] (1130) at (-0.5, -1.5) {};
		\node [style=none] (1131) at (-0.25, -1.5) {};
		\node [style=none] (1132) at (-1.25, 2.25) {$b'$};
		\node [style=none] (1133) at (-1.25, -2.25) {$b$};
		\node [style=none] (1134) at (-0.5, 2.25) {$x_{i < k} '$};
		\node [style=none] (1135) at (0.5, 2.25) {$x_k '$};
		\node [style=none] (1136) at (1.25, 2.25) {$x_{i > k} '$};
		\node [style=none] (1137) at (-0.5, -2.25) {$x_{i < k} $};
		\node [style=none] (1138) at (0.5, -2.25) {$x_k $};
		\node [style=none] (1139) at (1.25, -2.25) {$x_{i > k} $};
		\node [style=none] (1140) at (-0.75, 1.25) {$S(\phi_{(i)})^u$};
		\node [style=none] (1141) at (-0.75, -1.25) {$S(\phi_{(i)})^d$};
	\end{pgfonlayer}
	\begin{pgfonlayer}{edgelayer}
		\draw [in=90, out=-90] (1076.center) to (1077.center);
		\draw (1093.center) to (1094.center);
		\draw (1094.center) to (1095.center);
		\draw (1095.center) to (1096.center);
		\draw (1096.center) to (1093.center);
		\draw (1081.center) to (1074.center);
		\draw [in=90, out=-90] (1059.center) to (1102.center);
		\draw [in=270, out=90] (1078.center) to (1100.center);
		\draw (1108.center) to (1105.center);
		\draw (1105.center) to (1106.center);
		\draw (1106.center) to (1107.center);
		\draw (1107.center) to (1108.center);
		\draw (1111.center) to (1110.center);
		\draw (1113.center) to (1112.center);
		\draw (1119.center) to (1116.center);
		\draw (1116.center) to (1117.center);
		\draw (1117.center) to (1118.center);
		\draw (1118.center) to (1119.center);
		\draw (1122.center) to (1121.center);
		\draw (1126.center) to (1127.center);
		\draw [in=90, out=-90] (1124.center) to (1125.center);
		\draw [in=90, out=-90] (1131.center) to (1129.center);
		\draw (1128.center) to (1130.center);
	\end{pgfonlayer}
\end{tikzpicture} . \]
It is routine to see the structural maps of enrichment are all single-party representable, and that the cotensor is single-party representable (it is a family of identity functions), since all aforementioned maps simply represent wires into-which gaps have been cut. As a result, the single-party representable supermaps form a higher-order circuit theory. 
\end{example}

\begin{example}
Given any higher-order circuit theory $(\mathcal{P}, \mathcal{C})$ we can quotient it to impose that morphisms are equal so long as they return equal results on all possible input states. More precisely, given $S , T \in \mathcal{P}(\underline{a},\underline{b})$ with $\underline{a} = \star_i a_i $, we say that $S \cong T$ is and only if for every family of polymorphisms $\phi_i : \bullet \rightarrow a_i [x_i $ then the cuts of $S,T$ along the $\phi_i$ along the $a_i$ are equal, I.E, if and only if for all $\phi_i$ then $S \circ (\phi_1 ,  \dots ,  \phi_n) = T \circ {\phi_1 ,  \dots , \phi_n}$. 
It is easy to check that this quotient is compatible with polycategory composition, that is, if $S \cong T$ and $S' \cong T'$ then $S \circ T \cong S' \cong T'$.
Consequently, one can simply construct a new polycategory $\mathcal{P}^{\#}$ with objects given by those of $\mathcal{P}$ and with polymorphisms given by equivalence classes of those in $\mathcal{P}$. It is routine to see that if the original $\mathcal{P}$ forms part of a higher-order circuit theory $(\mathcal{P}, \mathcal{C})$, then all of the relevent structural morphisms for enrichment, cotensors, and the laws between them, are inherited to $\mathcal{P}^{\#}$ to form a new higher-order circuit theory $(\mathcal{P}^{\#}, \mathcal{C})$.
\end{example}

There is an obvious omission to this list of examples, which is the polycategory of coend optics \cite{hefford_coend}. This polycategory is manifestly non-symmetric, due to the aligning of the order of holes with the order of objects in the input list of any polymorphism. Whilst one cannot in general freely add symmetry to any polycategory, it is quite clear that coend optics are an example of a polycategory which can be symmetrized. Nonetheless, this example presents possibly the most obvious problem with the definition of higher-order circuit theories, that what ought to be a canonical example is not naturally and example. This suggests that really there is a piece of mathematical technology missing here from the literature which would make a cleaner formulation of higher-order circuit theories, namely, the theory of monoidal enrichment into non-symmetric, or possibly BV \cite{blute_BV}, polycategories. 


\section{A Convenient Notation for Cotensors}
In order to prove the main result of this paper, we will first introduce an additional diagrammatic function-box style notation which will make it easier for us to work with higher-order circuit theories. The key idea is to write the inverse of the isomorphism given by representability diagrammatically as a merging map: 
  \[   \tikzfigscale{0.8}{figs/bold_notation_1} ,    \]
with the defining diagramamtic condition being that merging and then splitting does nothing as follows:
  \[   \tikzfigscale{0.8}{figs/bold_notation_2}  \quad = \quad  \tikzfigscale{0.8}{figs/bold_notation_3} .    \]
It is likely, that such diagrams are related to, or more easily represented by, proof nets for linearly distributive categories \cite{Cockett1999LinearlyFunctors}. 
\subsection{Basic Features of Function-Boxes for Higher-Order Circuit Theories}
There are a few key diagrammatic laws which can then be directly inferred, and which will be used in the main theorem of this section, we will collect them together now to streamline the presentation of this theorem. First of all, since
    \[   \tikzfigscale{0.8}{figs/bold_timesdot_4}  \quad = \quad  \tikzfigscale{0.8}{figs/bold_timesdot_3}  \quad = \quad  \tikzfigscale{0.8}{figs/bold_timesdot_2} , \] then by the Frobenius law   \[  \tikzfigscale{0.8}{figs/bold_timesdot_4} \quad = \quad  \tikzfigscale{0.8}{figs/bold_timesdot_1}  ,  \] and so   \[   \tikzfigscale{0.8}{figs/bold_timesdot_5}  \quad = \quad  \tikzfigscale{0.8}{figs/bold_timesdot_6}  .   \]
The next rule, shows that black dots can be merged in the following sense:  \[   \tikzfigscale{0.8}{figs/bold_crunch1}  \quad = \quad  \tikzfigscale{0.8}{figs/bold_crunch_2} .    \] Indeed, this property  follows immediately from the fact that \[   \tikzfigscale{0.8}{figs/bold_crunch_3}  \quad = \quad  \tikzfigscale{0.8}{figs/bold_crunch_4}   \quad = \quad  \tikzfigscale{0.8}{figs/bold_crunch_5}   .   \]
  
Using the same technique, one can derive the 
symmetry law for function-box dots:
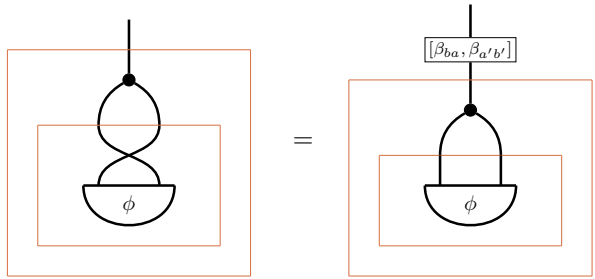
\[   \tikzfigscale{0.8}{figs/bold_symmetry_3}  \quad = \quad  \tikzfigscale{0.8}{figs/bold_symmetry_4}   .  \] The generalisation of this law which we proved for splitting maps, is also inherited to the corresponding function-box dots to give \[  
 \tikzfigscale{0.8}{figs/bold_braid_flip_1}  \quad = \quad  \tikzfigscale{0.8}{figs/bold_braid_flip_2}  .   \] 
Using these diagrammatic notations will greatly simplify the presentation of our main result, where we will see that in any theory of supermaps, the morphisms of type $[A,A'] \rightarrow [B,B']$ define locally-applicable transformations of the same type. No theorem on the soundness for even string diagrams for symmetric polycategories, let-alone the introduction of function boxes for representability, is known to the author. The proofs and calculations of this chapter should therefore be seen as sketched instructions for constructing more elaborate algebraic proofs.

\section{Upper Bound for Higher-Order Circuits}
In this section we give a basic upper bound on the possibilities for higher-order circuits over a fixed base theory $\cat{C}$ of lower order processes. The upper bound will be given by the theory of strong profunctors \cite{PARE199833, Tambara2006, pastro2007doublesmonoidalcategories}, which is a generalisation of the story of locally-applicable transformations given in the examples section \cite{hefford_supermaps}. We give a full description (albeit assuming strictness of monoidal structure) of the definition of strong profunctors here for completeness. 
\begin{definition}
An endoprofunctor on $\mathcal{C}$ is a functor $ \mathscr{T}\colon \mathcal{C}^{op}\times \mathcal{C}\to \mathbf{Set}$.
\end{definition}
In the following we give the specialisation of strong profunctors to the case in which they act upon monoidal categories which are strict. 
\begin{definition}
A strong profunctor on $\mathcal{C}$ is a profunctor on a strict monoidal category $\mathcal{C}$ equipped with a family of functions $ \alpha_y \colon \mathscr{T}(x,x') \to \mathscr{T}(xy,xy') $ natural in $x,x'$ and dinatural in $ y$, which additionally satisfy $\alpha_i = i$ and
\[
\begin{tikzcd}[row sep=large, column sep=large]
\mathscr{T}(a,b) \arrow[r, "\alpha_y"] \arrow[dr, "\alpha_{yz}"'] &
\mathscr{T}(na,nb) \arrow[d, "\alpha_z"] \\
& \mathscr{T}(xyz,x'yz)
\end{tikzcd}
\]
\end{definition}
An important theorem on the theory of strong profunctors, is that is it in-fact a presheaf category. More precisely we have that \cite{pastro2007doublesmonoidalcategories}: \[  \mathbf{StrProf}[\mathcal{C}] \cong \mathbf{copsf}[\mathbf{Optic}[\mathcal{C}]].  \]
As a result, the category $\mathbf{StrProf}[\mathcal{C}]$ has a monoidal structure induced by the Day tensor product.
\begin{definition}
The multicategory of strong profunctors has objects given by strong profunctors and morphisms $\mathscr{P}_1 \dots \mathscr{P}_n  \rightarrow \mathscr{Q}$ given by morphisms of type $\otimes_i \mathscr{P}_i \rightarrow \mathscr{Q}$ in $\mathbf{StrProf}[\mathcal{C}]$. 
\end{definition}
One can prove a concrete form for multi-morphisms in the multi-category induced by this tensor product, as noted in \cite{hefford_supermaps}, such multi-morphisms are naturally equivalent to specifying natural transformations of the type $[(\times_i \mathbf{Optic}[\mathcal{C}]) , \mathbf{Set}](\times_i \mathscr{P}_i({-}_i) , \mathscr{Q}(\otimes_i ({-}_i)))$, that is, families of functions \[  S_{x_{(i)}x_{(i)}'}  : \times_i \mathscr{P}_i(x_i , x_i') \rightarrow \mathscr{Q}(\otimes_i x_i , \otimes_i x_i'), \]
which are natural on optics. One can reconsider such transformations inductively, stating that a familly of functions $ S_{x_{(i)} x_{n+1} x_{(i)}' x_{n+1}'}  : (\times_i \mathscr{P}_i(x_i , x_i')) \times  \mathscr{P}_{n+1}(x_{n+1} , x_{n+1}' ) \rightarrow \mathscr{Q}((\otimes_i x_i) \otimes x_{n+1} , (\otimes_i x_i') \otimes x_{n+1})$ is a multimorphism $\mathscr{P}_1 \dots \mathscr{P}_n \mathscr{P}_{n+1}  \rightarrow \mathscr{Q}$ if and only if 
\begin{itemize}
\item For every choice of $\phi_1 \dots \phi_n$ then the induced familly of functions $S^{\phi_{1} \dots \phi_n}_{x_{n+1} x_{n+1}' }$ is a multimorphism of type $ \mathscr{P}_{n+1}  \rightarrow \mathscr{Q}$,
\item For every $\phi_{n+1} \in \mathscr{P}_{n+1}(x_{n+1} , x_{n+1}' )$ then the induced familly of functions $\mathscr{Q}(\beta_{\otimes_i x_i , x_{n+1}}  , \beta_{\otimes_i x_i , x_{n+1}} ) \circ S^{\phi_{n+1}}_{x_{(i)} x_{(i)}' }$ is a multimorphism of type $\mathscr{P}_1 \dots \mathscr{P}_n  \rightarrow \mathscr{Q}$.
\end{itemize}
We now prove our main result, that any higher-order circuit theory embeds within the previously put-forward categorical construction for higher-order processes based on strong profunctors \cite{hefford_supermaps}. 
\begin{theorem}
For every higher-order circuit theory $(\mathcal{P} , \mathcal{C})$ there exists a faithul multifunctor \[\mathcal{F}: \mathcal{P}^{\#} \rightarrow \mathbf{StrProf}[\mathcal{C}]. \]
\end{theorem}
\begin{proof}
On objects we take $\mathcal{F}(a)(x,x') = \mathcal{P}(\bullet , a [x,x'])$ and on morphisms we take \[ \mathcal{F}(S: a_1 \dots a_n  \rightarrow b)_{e_1 \dots e_n ,e_1 ' \dots e_n '}(\phi_1 ,  \dots , \phi_n)  \] to be given by: \[  \tikzfigscale{0.8}{figs/multifunctor_def_1arx}.  \] 

First, we confirm that the assignment $\mathcal{F}(a)(x,x') = \mathcal{P}(\bullet , a [x,x'])$ indeed returns a strong profunctor. On morphisms we define \[  \mathcal{P}(\bullet , a [f,g]) \left(  \tikzfigscale{0.8}{figs/prof_0a}  \right)   \quad := \quad \tikzfigscale{0.8}{figs/prof_0b}, \]
with functoriality given by
\[  \tikzfigscale{0.8}{figs/prof_1} \quad = \quad    \tikzfigscale{0.8}{figs/prof_2}. \] The strength is given by the function  \[  \alpha_{y} \left(  \tikzfigscale{0.8}{figs/sprof_0a}  \right)   \quad := \quad \tikzfigscale{0.8}{figs/sprof_0b}, \] where naturality and dinaturality of the strength are given by \[  \tikzfigscale{0.8}{figs/sprof_nat_1} \quad = \quad    \tikzfigscale{0.8}{figs/sprof_nat_2},  \]  and 
\[  \tikzfigscale{0.8}{figs/sprof_dinat_1} \quad = \quad    \tikzfigscale{0.8}{figs/sprof_dinat_2},  \] 
respectively. 
The associativity and unit laws for strong profunctors are then given by \[  \tikzfigscale{0.8}{figs/assoc_1} \quad = \quad    \tikzfigscale{0.8}{figs/assoc_2} \quad = \quad    \tikzfigscale{0.8}{figs/assoc_3} , \]
and \[  \tikzfigscale{0.8}{figs/sprof_unit_1} \quad = \quad    \tikzfigscale{0.8}{figs/sprof_0a} ,  \] respectively.  

We now check the result is indeed a locally-applicable transformation of the same type, to do so we check commutation with all combs by separately checking commutation with sequential and parallel compositions\footnote{See the proof of the theorem on slide drag for a fuller discussion of sliding and dragging rules.}. We give the proof for one-input supermaps for simplicity, the multi-input case can then be proven by induction. In the single-input case, first note that the functor simplifies to  \[  \tikzfigscale{0.8}{figs/simple_1} \quad = \quad    \tikzfigscale{0.8}{figs/simple_2}, \]
and so we have that the dragging law and sliding law are inherited by the associativity of polycategory composition \[  \left( \tikzfigscale{0.8}{figs/boldsingle_2arx} \quad = \quad    \tikzfigscale{0.8}{figs/boldsingle_4arx} \right)  \quad \bigwedge  \quad  \left(   \tikzfigscale{0.8}{figs/slide_1_arx} \quad = \quad    \tikzfigscale{0.8}{figs/slide_2_arx} \right) .  \] Let us now verify the multi-input case by induction, assume that the hypothesis is true for the $N$-input case, now consider the $N+1$ input case. A family of functions is an $N+1$-input locally-applicable transformation if filling in the first $N$ entries gives an $N=1$ input locally-applicable transformation and filling in the last entry gives an $N$-input locally applicable transformation (up to applying swaps). Indeed, for the former case see that since  \[   \tikzfigscale{0.8}{figs/boldinduction1arx} \quad = \quad   \tikzfigscale{0.8}{figs/boldinduction2arx},   \]
we can use representability to recover the required form
 \[   \tikzfigscale{0.8}{figs/bold_collapse_1arx} \quad = \quad   \tikzfigscale{0.8}{figs/bold_collapse_2arx} ,  \]
 and so filling in the first $N$-inputs gives a $1$-input locally applicable transformation. For the latter case  \[   \tikzfigscale{0.8}{figs/newinduction} \quad = \quad   \tikzfigscale{0.8}{figs/boldinduction3arx} , \] 
 
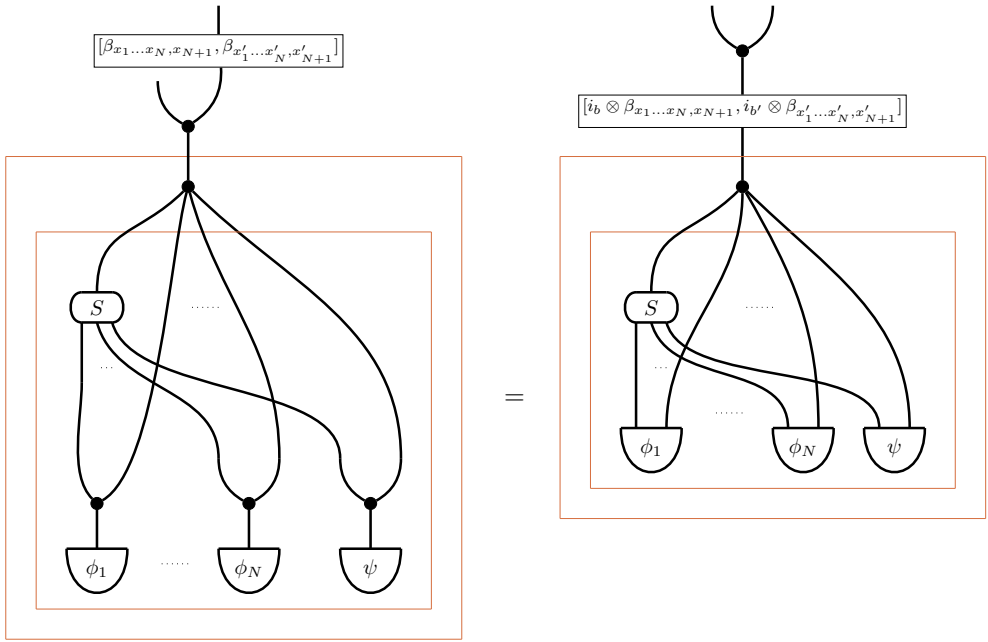
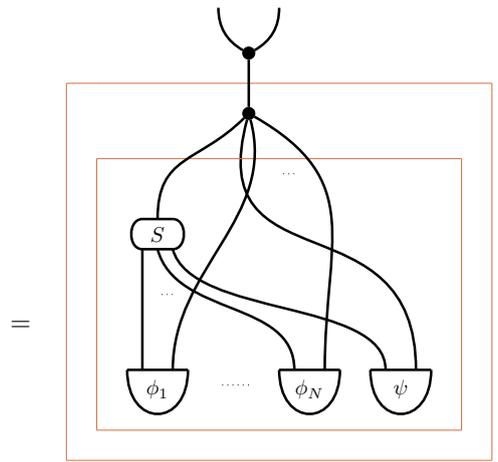
 \[   = \quad   \tikzfigscale{0.8}{figs/boldinduction4arx}, \] which by the interchange law can be rewritten as    \[   \tikzfigscale{0.8}{figs/boldinduction5arx}  \quad = \quad  \tikzfigscale{0.8}{figs/boldinduction6arx}. \] After using the merging rule this gives \[  \tikzfigscale{0.8}{figs/boldinduction7arx}  ,   \] and so filling input $N+1$ gives up to swaps an $N$-input locally applicable transformation by assumption of the induction hypothesis, completing the proof. Functoriality in the single input case is clear.
Multi-functoriality, is verified as follows \[    \tikzfigscale{0.8}{figs/multi_func_law_1arx}  \]
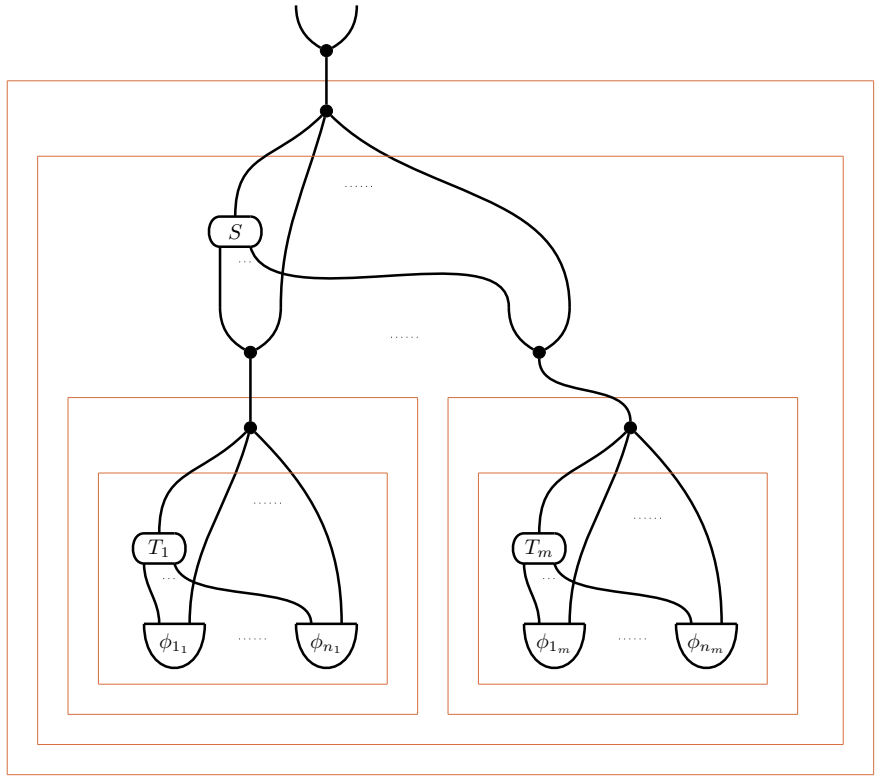
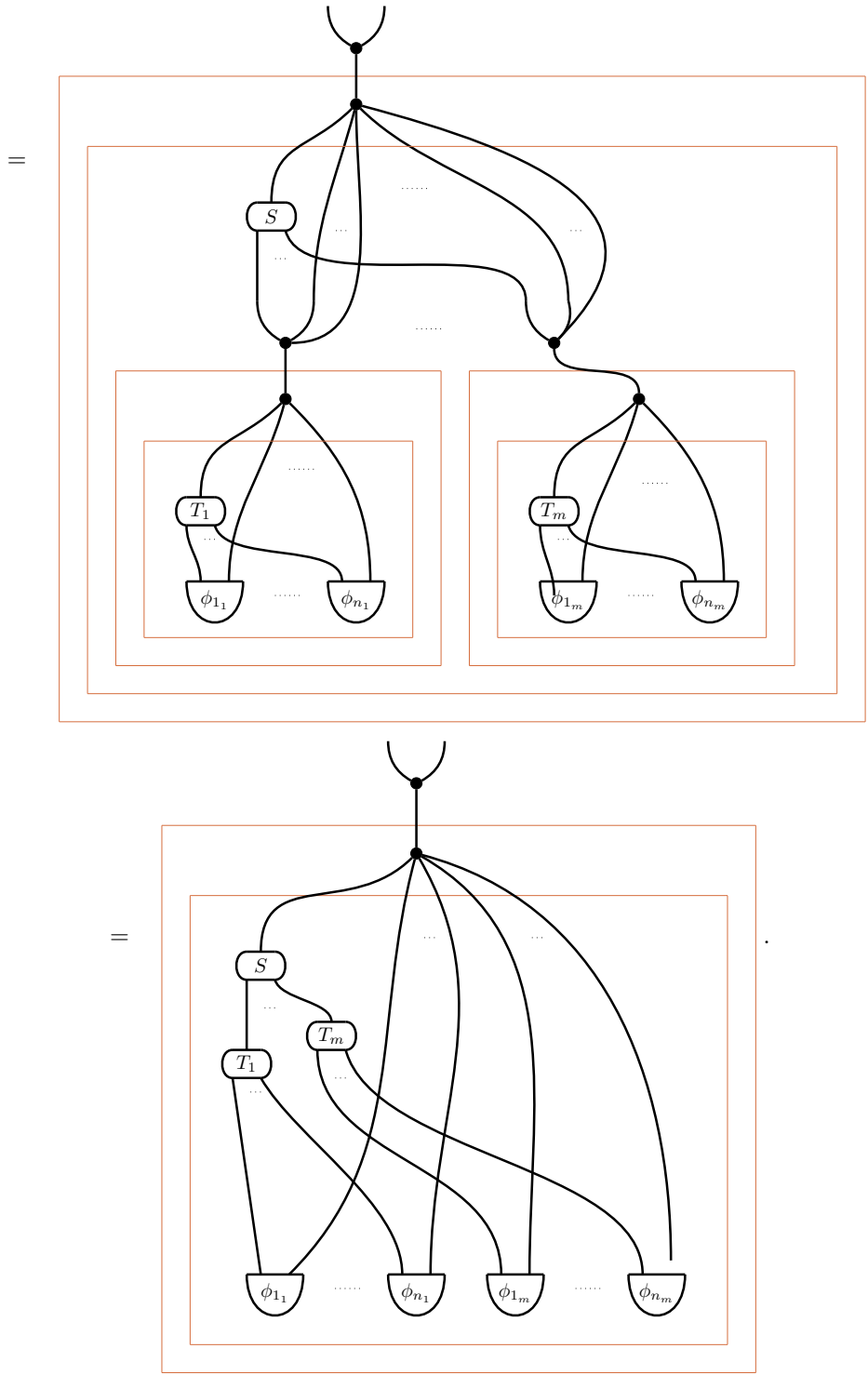
 \[ = \quad  \tikzfigscale{0.8}{figs/multi_func_law_2arx}    \] 
  \[    = \quad  \tikzfigscale{0.8}{figs/multi_func_law_3arx}.     \]
\end{proof}

\section{Summary}
Higher-order circuits, albeit in the case in which all possible compositional structures are treated as strict, can be given a relativity simple categorical semantics in terms of polycategories, enrichment, and cotensors. Each law for higher-order circuits has an interpretation as a graphically tautology, and with only those laws in place it is possible to make a direct connection between higher-order circuits and previous categorical approaches to the study of holes based on profunctor optics. As a result, we appear to have distilled an intuitive concept into categorical algebra. Furthermore, we have provided a formal setting on top of which resource theories of higher-order processes might be studied as sub higher-order circuit theories.

Natural questions on the theoretical side are:
\begin{itemize}
\item How can the definition of higher-order circuits be adapted to the more complicated case in which $\cat{C}$ is no-longer forced to be strict. In particular, could there be strictification and coherence theorems \cite{Lane1971CategoriesMathematician} for higher-order circuits suitably categorified? 
\item Can the planar diagrams with holes be used to build a sound and complete graphical language for higher-order circuit theories, in the same way that string diagrams are sound and complete for monoidal categories \cite{Joyal1991TheI}?
\item Do there remain other reasonable laws for planar holes (ones which can be interpreted as graphical tautologies) which are not generated by the laws put forward here?
\item Higher-order quantum operations \cite{Bisio_2019, kissinger_caus} can be equipped with rather complex composition operations \cite{apadula2022nosignalling,simmons_completelogic,  SimmonsKissinger2022,hefford2025bvcategoryspacetimeinterventions}. Most notably in the introduction of a sequencing tensor product which gives higher-order operations the structure of BV-categories in the representable case \cite{blute_BV}. It is unclear at this stage whether information about the sequencing tensor product has been completely lost, or whether its non-representable structure might be encoded into the structure of enrichment operations (in particular the sequential composition morphisms).
\item In this article we focussed on upper bounds, however, it is quite apparent that their ought to exist lower bounds for higher-order circuits too. Most notably, the existence of enrichment morphisms allows for the construction of all combs \cite{wilson2023mathematical}. Naturally then, one can naturally expect that combs or more generally coend-optics over a category might embed canonically into all other higher-order circuit theories over the same category. The existence of such embeddings and their potential uniqueness as structure preserving maps between higher-order circuit theories is left as a future topic of investigation.
\item From the perspective of the foundations of physics, the beyond-monoidal structure of higher-order processes motivates the generalisation of process theories \cite{selby2025generalisedprocesstheories}, in terms of operads over wiring diagrams \cite{yau2018operads, spivak_normal_form}. It is natural then to wonder, whether there might exist an operad over which the algebras are higher-order circuit theories.  
\end{itemize}
The answers to these questions will determine whether there might exist in the most formal possible sense a stable categorical framework for holes and higher-order processes.

\section{Acknowledgements}
The majority of this work was completed while the author was a PhD student at the Quantum Group of the Computer Science Department at the University of Oxford, and can be found in the associated thesis  \cite{wilson2023compositional}. Refinements have been made since, both while affiliated to University College London and now CentraleSupelec, University of Paris-Saclay.
The author (MW) is particularly grateful to Giulio Chiribella (GC) and James Hefford (JH), this work puts together many ideas produced through ongoing conversations and projects in collaboration with both GC and JH. From GC in particular, comes the idea that multi-input higher-order processes could be modelled using theories with preferred sequential and parallel composition \textit{morphisms} \cite{wilson2023mathematical}; and from JH comes the idea that strong profunctors could provide the domain for locally-applicable transformations and so generalise the spaces on which they act \cite{hefford_supermaps}. 
MW is also grateful to Peter Selinger, Gaurang Agrawal, and Benoit Valiron, for useful conversations during the development of this manuscript. 
MW was funded by the Engineering and Physical Sciences Research Council grant number EP/L015242/1 while at the University of Oxford and grant number EP/W524335/1 while at University College London. This project/publication was made possible through the support of the ID$\#$62312 grant from the John Templeton Foundation, as part of “The Quantum Information Structure of Spacetime” Project (QISS). The opinions expressed in this project/publication are those of the author(s) and do not necessarily reflect the views of the John Templeton Foundation. This work has also been funded by the French National Research Agency (ANR), (i) by the project TaQC ANR-22-CE47-0012, and (ii) within the framework of “Plan France 2030”, under the research projects EPIQ ANR-22-PETQ-0007, OQULUS ANR-23-PETQ-0013, HQI-Acquisition ANR-22-PNCQ-0001 and HQI-R$\&$D ANR-22-PNCQ-0002.



\appendix

\section{The Multi-Party Braid Law}

\begin{lemma}
In any symmetric higher-order circuit theory, the following law holds \[   \tikzfigscale{0.8}{figs/bold_braid_1}  \quad = \quad  \tikzfigscale{0.8}{figs/bold_braid_2}.   \]
\end{lemma}
\begin{proof}
Imagine that the property were true for some $n$, then when checking the property for $n+1$ we find that \[   \tikzfigscale{0.8}{figs/bold_braid_proof_1}  \quad = \quad  \tikzfigscale{0.8}{figs/bold_braid_proof_2} ,    \] which after using naturality and associativity gives

\[  = \quad  \tikzfigscale{0.8}{figs/bold_braid_proof_3}  \quad = \quad  \tikzfigscale{0.8}{figs/bold_braid_proof_4} .     \] 
Next, using the property for $n=1$ along with naturality and rules for the composition of swap morphisms in $\mathbf{C}$ gives
\[  = \quad  \tikzfigscale{0.8}{figs/bold_braid_proof_5}  \quad = \quad  \tikzfigscale{0.8}{figs/bold_braid_proof_7}      \] 
which by associativity gives the required form
\[   = \quad  \tikzfigscale{0.8}{figs/bold_braid_proof_8} .    \] 
Note that as a consequence of these property we can also see that \[  
 \tikzfigscale{0.8}{figs/bold_braid_extra_1}  \quad = \quad  \tikzfigscale{0.8}{figs/bold_braid_extra_2}      \] 

\end{proof}
\end{document}